\documentclass[11pt]{article}

\pdfoutput=1

\usepackage{fullpage}
\usepackage{amssymb}
\usepackage{amsmath}
\usepackage{amsfonts}
\usepackage{amsthm}
\usepackage{url}
\usepackage{graphicx}
\usepackage{subfig}
\usepackage{fancyhdr}
\usepackage{url}

\usepackage[utf8]{inputenc}
\usepackage[english]{babel}
\usepackage{hyperref}

\newcommand{\nc}{\newcommand}
\nc{\heading}[1]{\begin{center} \large \bf #1 \end{center}}

\newcommand{\oA}{\overline{A}}
\newcommand{\oB}{\overline{B}}
\newcommand{\oI}{\overline{I}}

\newcommand{\oR}{\overline{R}}
\newcommand{\oS}{\overline{S}}

\newcommand{\Ex}{\mathsf{E}}

\newcommand{\horizline}{\vspace{2pt} \noindent \rule{\textwidth}{0.5pt}}

\theoremstyle{plain}
\newtheorem{prop}{Proposition}

\newtheorem{definition}{Definition}

\usepackage{glossaries}
\makeglossaries
\newglossaryentry{latex}
{
    name=latex,
    description={Is a mark up language specially suited
    for scientific documents}
}
\newglossaryentry{maths}
{
    name=mathematics,
    description={Mathematics is what mathematicians do}
}
\begin{document}

\date{May 28, 2020}
\title{Stochastic Modeling of an Infectious Disease \\
 Part I: Understand the Negative Binomial Distribution\\
 and Predict an Epidemic More Reliably}
\author{Hisashi Kobayashi\footnote{The Sherman Fairchild University Professor of Electrical Engineering and Computer Science, Emeritus.~~
Email: hisashi@princeton.edu, Blog: \url{http://hp.hisashikobayashi.com}, Wikipedia: \url{https://en.wikipedia.org/wiki/Hisashi_Kobayashi}
 }\\
   Dept. of  Electrical Engineering \\
   Princeton University \\
   Princeton, NJ 08544, U.S.A.}

\maketitle

\begin{center}\textbf{Summery}\end{center}

Why are the epidemic patterns of COVID-19 so different among different cities or countries which are similar in their populations, medical infrastructures, and people's behavior?  Why are forecasts or predictions made by so-called experts often grossly wrong, concerning the numbers of people who get infected or die? 

The purpose of this study is to better understand the stochastic nature of an epidemic disease such as COVID-19, and answer the above questions.  The author hopes that this article will provoke discussions among the ``modeling communities" and stimulate mathematically inclined people to study this interesting and important field, i.e., mathematical epidemiology.

Much of the work on infectious diseases has been based on ``SIR deterministic models," pioneered by Kermack and McKendrick in 1927.  In our study we will explore several stochastic models that can capture the essence of the seemingly erratic behavior of an infectious disease, which the deterministic model cannot explain.  A stochastic model, in its formulation,  takes into account the random nature of an infectious disease.  Thus, such a model, if properly constructed, should be able to provide a more reliable and informative forecast of an epidemic pattern.  

The stochastic model we study in this article is based on the \textbf{\emph{birth-and-death process with immigration}}(BDI for short), which was originally proposed in the study of population growth or extinction of some biological species. To the best of this author's knowledge, the BDI process model has not been investigated by the epidemiology community, perhaps for the reason we briefly discuss in Section 2. 

The general birth-and-death (BD) process usually defies an attempt to obtain a closed solution for the time-dependent (i.e., transient) probability distribution of the population size, etc.  The BDI process, however, is among a small number of BD processes, which we can solve analytically.  An important feature of the BDI process is that its probability distribution function is a generalized \textbf{\emph{negative binomial distribution}} (NBD), with its parameter $r$ being less than one.  We show that the ``coefficient of variation" (the standard deviation normalized by the mean) of the BDI process is larger than $r^{-1}>1$. Furthermore, a NBD with small $r$ has a long tail in its distribution form, like the zeta distribution (aka Zipf's law).  These properties of the infection process explain why 
actual infection patterns exhibit enormously large variations.  Furthermore, the mean value of the number infected provided by a deterministic model is far from the median of the distribution.  This explains why any forecast or prediction based on a deterministic model will fail more often than not.

In Part II of this report \cite{kobayashi:2020b}, we will present results of our extensive simulation study and further analysis of the stochastic model based on the BDI process.   

\paragraph{\em Keywords:}
Infectious disease, COVID-19, Forecast and prediction, Stochastic model, Deterministic model, Kermack-McKendrick's SIR model, Basic and effective reproduction numbers, Birth-and-death process with immigration (BDI), Probability generating function (PGF), Partial differential equation (PDE),   Negative binomial distribution (NBD), Coefficient of variation (CV),  Compound Poisson process, Fisher's logarithmic distribution.

\tableofcontents

\section{Introduction}
Most of the mathematical models of infectious diseases seem to be based on the Kermack-McKendrick model published in 1927~\cite{kermack-mckendrick:1927}, which was proposed to explain the rapid rise and fall in the number of infected population observed in epidemics such as the great plague in London where more than 15\% of the population died  (1665-66);  and the cholera outbreak in London caused by contamination in the Thames River (1865), and the plague epidemic in Bombay (1906)~\cite{bacaer:2011}.  The model consists of a system of three coupled nonlinear ordinary differential equations for the infected population $I(t)$, the susceptible population  $S(t)$, and the recovered  population $R(t)$.  Kermack-McKendrick's \textbf{SIR model} is a \textbf{deterministic model}, which provides the \textbf{expected values} of these processes, which we denote as $\oS(t)(=\Ex[S(t)]$, $\oI(t)(=\Ex[I(t)]$ and $\oR(t)(=\Ex[R(t)])$. A majority of biological and epidemiological models~\cite{anderson-may:1979,anderson-may:1991,martcheva:2010} fall in this class of deterministic models. 

Actual observed data of the infected population, for instance,  is merely an instance or \textbf{a sample path} of this stochastic process $I(t)$. The process naturally deviates from the expected value $\oI(t)$ obtained by a deterministic model. Thus, a deterministic model alone fails to provide any quantitative explanation when observed data  differ significantly from the expected value.  

In a \textbf{stochastic} (or \textbf{probabilistic}) model, on the other hand, the intrinsic stochastic nature of a process is explicitly taken into account in its model formulation. 
The importance of stochastic processes in relation to problems of population growth was pointed out by W. Feller in 1939 \cite{feller:1939}.  He considered the \textbf{birth-and-death} process in which the expected birth and death rates (per person per unit time) were constants, say, $\lambda$ and $\mu$. D. G. Kendall \cite{kendall:1948a} extended Feller's birth-and-death (BD) process  by considering the birth and death rates as any specified functions of the time $t$, $\lambda(t)$ and $\mu(t)$. The BD process is a special class of \emph{time-continuous discrete-state Markov process}, and has found applications in many scientific and engineering fields, including population biology \cite{feller:1968}, teletraffic and queueing theory \cite{syski:1986}, \cite{kleinrock:1975}, \cite{kelly:1979}, system modeling \cite{kobayashi:1978}\cite{kobayashi-mark:2008}, pp. 63-94, \cite{kobayashi-mark-turin:2012}, pp. 407-410. 

I have done some investigation, with help from Prof. Hideaki Takagi, whose unpublished lecture note at the University of Tsukuba \cite{takagi:2007} provided me with several references, as to who coined the term ``birth-and-death process with immigration," and have found the English statistician M.S. Bartlett (1910-2002) in his famous book ``An Introduction to Stochastic Processes" \cite{bartlett:1978} (1st edition in 1955) in his discussion of the birth-and-death process in ``Section 3.4 Multiplicative chain: subsection 3.4.1 ``The effect of immigration,''that he uses the phrase "a birth-death-and-immigration process."  But Bartlett's doctoral student, David G. Kendall (1918-2007) gives a detailed analysis of the BDI process in his 1949 article \cite{kendall:1948b}, which Bartlett refers to in his 1949 article \cite{bartlett:1949}. So my tentative conclusion was  that Kendall was the first that worked on the BDI although he did not use the term ``birth-death-immigration" or something to that effect.
 
Another English statistician, Norman T. J.  Bailey published in 1964 ``The Elements of Stochastic Process with Applications to the Natural Sciences," \cite{bailey:1964}, and  in ``Section 8.7: The effect of  immigration" (pp. 97-101), he gives a thorough treatment of the BDI process.  Linda J. S. Allen, ``An Introduction to Stochastic Processes with Applications to Biology," (2nd Edition, 2011) \cite{allen:2011} discusses the BDI process in Section 6.4.4: Simple Birth and Death with Immigration (pp. 254-258), but her focus  seems to be more on the stable case.   Frank P. Kelly \cite{kelly:1979} gives a brief treatment, providing the steady-state distribution.  
All other numerous textbooks on random processes make no mention of the BDI process, and even the above handful of authors who might have had epidemiologists in mind among their readership seem to treat the BDI process for its possible application to population biology, and none allude to its use in epidemiology.
   
\section{A Brief Review of SIR Deterministic Model}  \label{sec-SIR}

As stated earlier, a majority of mathematical models reported in the literature on infectious diseases  have been deterministic models, following the pioneering work of almost a century ago by W.G. Kermack and A.G. McKendrick \cite{kermack-mckendrick:1927} \footnote{William Ogilvy Kermack (1898-1970) was a Scottish biochemist and Anderson Gray McKendrick (1876-1943) was a Scottish military physician and epidemiologist.}.  They assumed 
\begin{enumerate}
\item At any time $t$, an individual is either susceptible (\emph{S}), infected and infectious (\emph{I}) or recovered and immune (\emph{R}).
\item Only susceptible individuals can get infected, remain infectious for some time, and recover and become completely immune.
\item There are no births, deaths, immigration or emigration during the study period.  In other words, the community is \emph{closed}.
\end{enumerate}

Consequently, individuals can only make two types of transitions:  (i) from \emph{S} to \emph{I}, and  (ii) from \emph{I} to \emph{R}. Thus, the Kermack-Mckendrick model is often referred to as an \emph{SIR} epidemic model.    A model which assumes no immunity (i.e., a recovered person becomes immediately susceptible) is called an \emph{SIS} model.  
If we explicitly consider an \emph{exposed} state, during which an infected individual is not yet infectious, the model is called an \emph{SEIR} model.  A model in which immunity wanes after some period is called an \emph{SIRS} model, and so forth.  In the remainder of this section we give a brief account of the SIR model in a closed community so that the reader can compare this deterministic model to our stochastic model to be presented in the next section.  For details of the Kermack-McKedrick type deterministic models,  the readers are referred to abundant books and articles;  Anderson and May, \cite{anderson-may:1991}, Martcheva \cite{martcheva:2010} to name just a few.

Let $\oS(t)$, $\oI(t)$ and $\oR(t)$, respectively\footnote{We adopt this notation to distinguish them from $S(t)$, $I(t)$ and $R(t)$ which are stochastic processes, as used in other parts of this article. The time functions used in deterministic models usually correspond to the expectation or stochastic mean of the corresponding stochastic processes.}, denote the number of the susceptible, infected and recovered at time $t$.  Since we assume no births, deaths, immigration nor emigration, we have
\begin{align}
\oS(t)+ \oI(t) + \oR(t)= N,~~\mbox{for all}~~t\geq 0,\label{N-const}
\end{align}
where $N$ is a constant number, representing the population of the community.

From the set of assumptions stated above, the deterministic processes can be defined by the following set of three differential equations:
\begin{align}
\frac{d\oS(t)}{dt}&=-\beta \oS(t)\oI(t),\label{KM-1}\\
\frac{d\oI(t)}{dt}&=\beta \oS(t)\oI(t)-\mu\oI(t),  \label{KM-2}\\
\frac{dR(t)}{dt}&=\mu\oI(t). \label{KM-3}
\end{align}
These differential equations, together with (\ref{N-const}) and the initial condition
\begin{align}
\oI(0)=I_0, ~~\mbox{and}~~\oR (0)=0,   \label{initial-cond}
\end{align} 
define the deterministic model.  It is easy to see that $\oS(t)$ is monotone decreasing, and $\oR(t)$ is monotone increasing.  The function $\oI(t)$ increases or decreases at time $t$, depending on whether the ratio $R_t$ defined by  the following expression is greater or smaller than unity. 
\begin{align}
{\cal R}_t=\frac{\beta\oS(t)}{\mu}\leq \frac{\beta\oS(0)}{\mu}={\cal R}_0  \label{def-R_0-in-KM}
\end{align}

Its initial value ${\cal R}_0$ is referred to as the \textbf{\emph{basic reproduction number}}, a term having its origin in demography.  The ratio $R_t$ is called the \textbf{\emph{effective reproduction number.}} or the \textbf{\emph{real-time reproduction number}} and is more meaningful than ${\cal R}_0$, which is a static number, in estimating the current epidemic situation and making a decision to control the epidemic.

The term $\beta\oS(t)\oI(t)$ in (\ref{KM-1}) and (\ref{KM-2}) comes from the argument that the susceptible must have contact with the infected in order to get infected, and if we assume some sort of \emph{uniform mixing},  the infections should occur at a rate proportional to $\oS(t)\oI(t)$.  Consequently, the unit of the parameter $\beta$ is [person/unit-time/person$\cdot$person], whereas the other parameter \footnote{Another Greek letter $\gamma$ is often used for the recovery rate, but we use $\mu$ to be consistent with our notation in the next section.} $\mu$ has the unit of [person/unit-time/person].

When the total population $N$  of (\ref{N-const}) is sufficiently large, $S(t)\approx N$ is much larger than $I(t)$ and can be treated as unchanged,  at least in the initial phase of an epidemic.  Then by setting
\begin{align}
\lambda \approx\beta S(t)\approx \beta S(0) ~[\mbox{person/unit-time}\cdot\mbox{person}] , \label{lambda-beta}
\end{align}
and by substituting this into (\ref{KM-2}), we have the following ordinary differential equation (ODE):
\begin{align}
\frac{d\oI(t)}{dt}=(\lambda-\mu)\oI(t),  \label{ODE-I(t)}
\end{align}
from which and the initial condition (\ref{initial-cond}), we readily find the solution for $\oI(t)$:
\begin{align}
\oI(t)=I_0 e^{at},~~t\geq 0,~~\mbox{where}~~a=\lambda-\mu.  \label{KM-I(t)}
\end{align}
which is an exponentially growing or decaying function, depending on whether $a>0$ or $a<0$.  When $a=0$, it is a constant $I_0$ for all $t\geq 0$.

When the infected $I(t)$ grows to the extent that the approximation (\ref{lambda-beta}) no longer holds, i.e., the ``infinite population" assumption fails, we have to deal with the nonlinear differential equations of (\ref{KM-1}) and (\ref{KM-2}).  A major advantage of the SIR model is that because of the product term $S(t)I(t)$, the differential equations take into account explicitly the fact that occurrences of infections will gradually decreases towards to zero as the susceptible population becomes extinct towards the end of the infection processes. The main drawback of the SIR model, on the other hand, is its inability to capture any probabilistic fluctuation of the infection process. The SIR model
may be an appropriate model in describing the interactions between the susceptible group and infected group in a closed environment, such as a hospital, retirement home, cruise ship, night club, etc.  But it is a poor model in describing a major outbreak of an epidemic in a larger environment such as a city or a country, where most infections take place independently and randomly, and the product term $\beta\oS(t)\oI(t)$ does not have any significant meaning.  Furthermore, the product term makes the entire SIR model a nonlinear system, and makes the system mathematically intractable, except for a few simple cases, which may not be useful in reality.

As we will show in this article and Part II \cite{kobayashi:2020b}, the deterministic model, in addition to being unable to describe the stochastic fluctuation of an epidemic pattern, is more likely to grossly overestimate the number of casualties. Thus, the deterministic model is not only limited in its applicability, but can be damaging and harmful in some cases.

\section{A Stochastic Model for an Infectious Disease} \label{sec-stochastic_model}
In this and following sections we will discuss our stochastic model based on the \textbf{\emph{birth-and-beath with immigration}} (BDI) process.  It is a special case of general birth-and-death (BD) process\footnote{William (Willy) Feller introduced what is now known as the ``birth-and-death'' process in his 1939 article \cite{feller:1939} published in German  regarding the problem of population growth.  He used the term ``Vermehrung" (reproduction) and ``Tod'' (death). David. G. Kendall in his 1948 paper \cite{kendall:1948a} referred to Feller's model as a ``birth-and-death" process. Feller used this term in his Volume I \cite{feller:1968}, whose first edition was in 1950.}.  Before we present a detailed mathematical analysis of this model, we will show an example of our simulation model based on the BDI
process.
\begin{figure}[ht]
  \begin{minipage}[b]{0.45\textwidth}
  \centering
  \includegraphics[width=\textwidth]{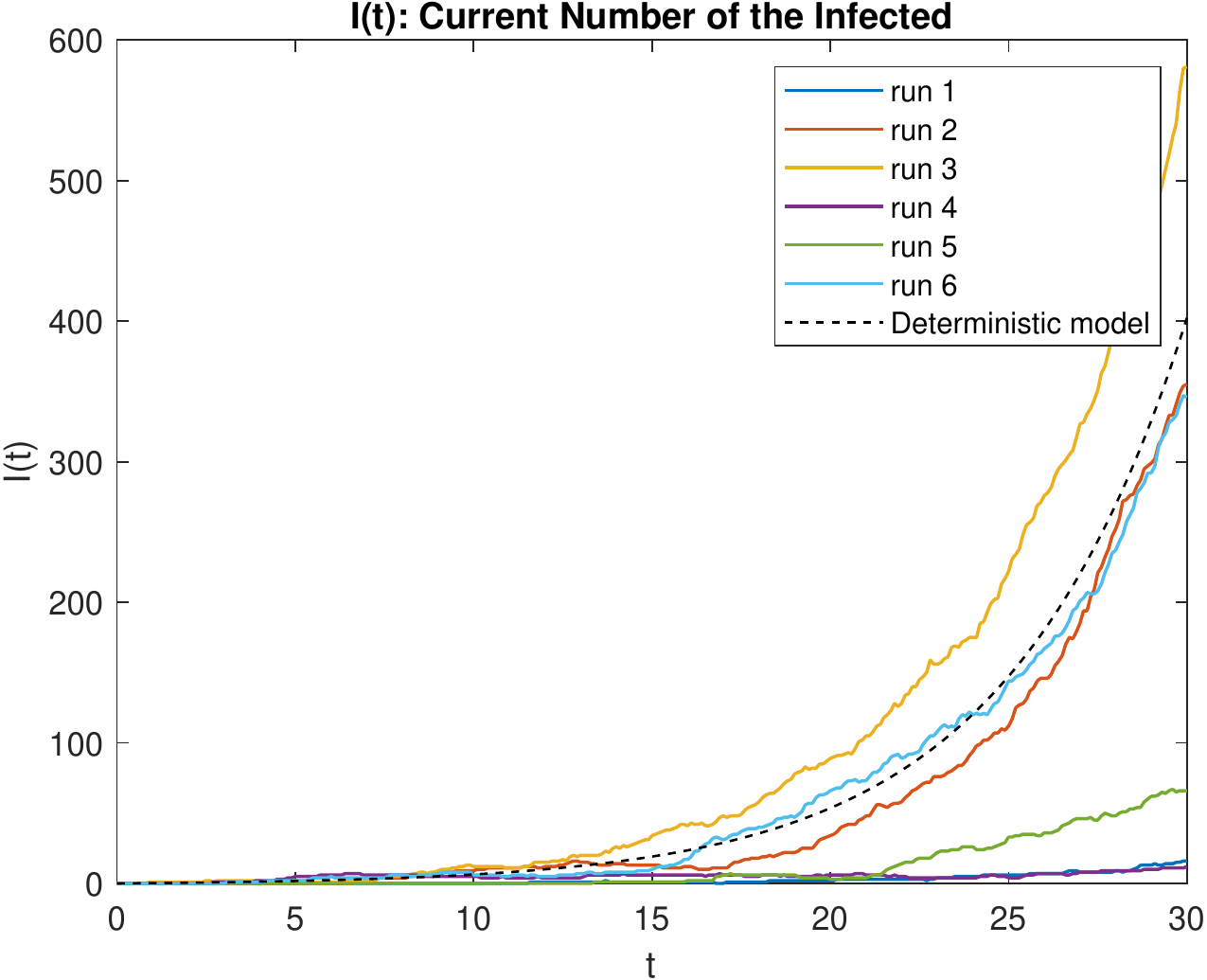}
  \caption{The first 6 simulation runs of BDI process and a deterministic model curve.}
  \label{fig:simu_1-6}
  \end{minipage}
  \qquad\qquad
  \begin{minipage}[b]{0.45\textwidth}
  \centering
  \includegraphics[width=1\textwidth]{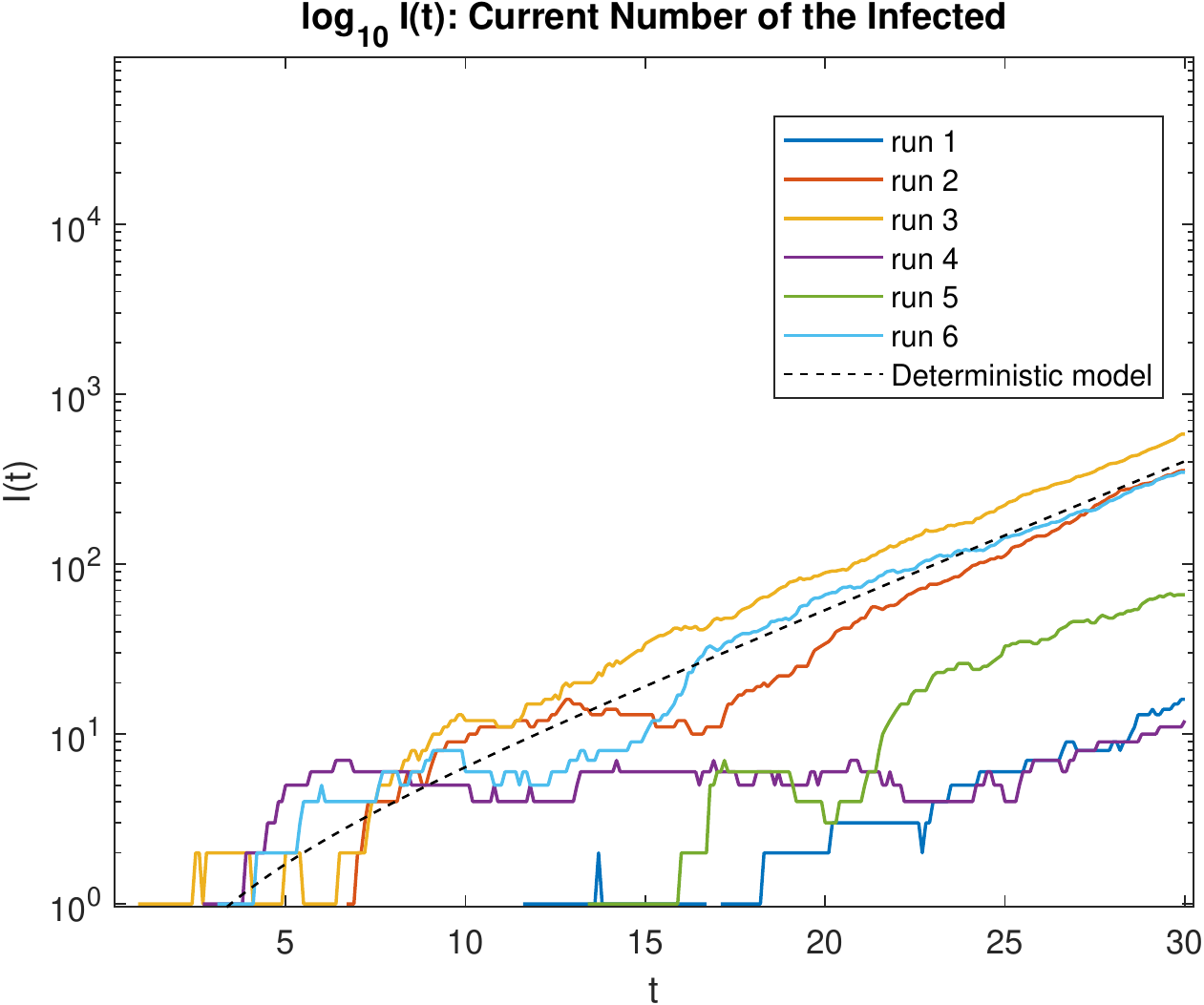}
  \caption{Semi-log plots of the same runs. The baseline of $y$ axis is $10^0=1$.} 
  \label{fig:Semilog_simu_1-6}
  \end{minipage}
\vskip 1em
  \begin{minipage}[b]{0.45\textwidth}
  \centering
  \includegraphics[width=1\textwidth]{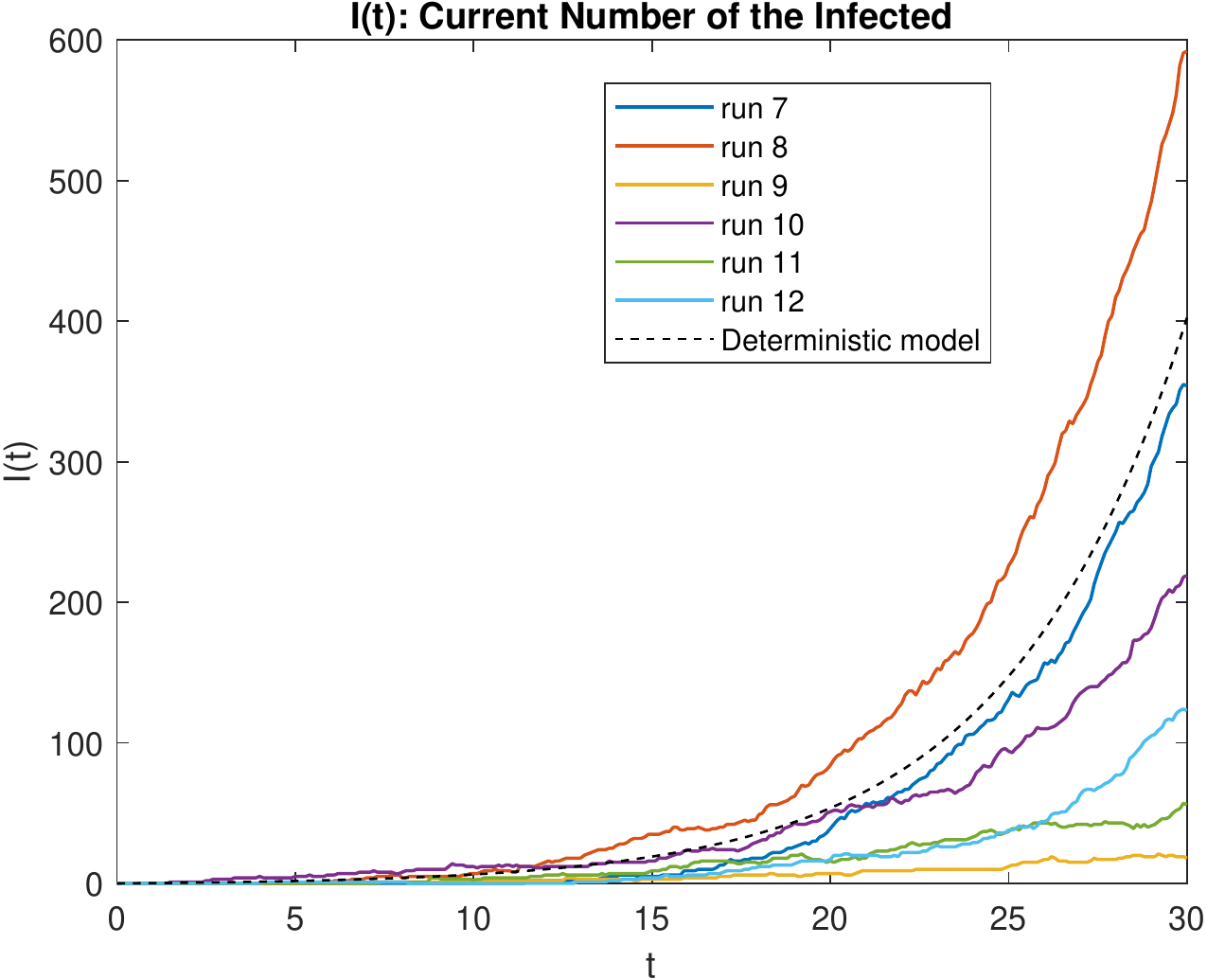}
  \caption{The next 6 simulation runs of BDI process and a deterministic model curve.}
  \label{fig:simu_7-12}
  \end{minipage}
  \qquad \qquad
  \begin{minipage}[b]{0.45\textwidth}
  \centering
  \includegraphics[width=1\textwidth]{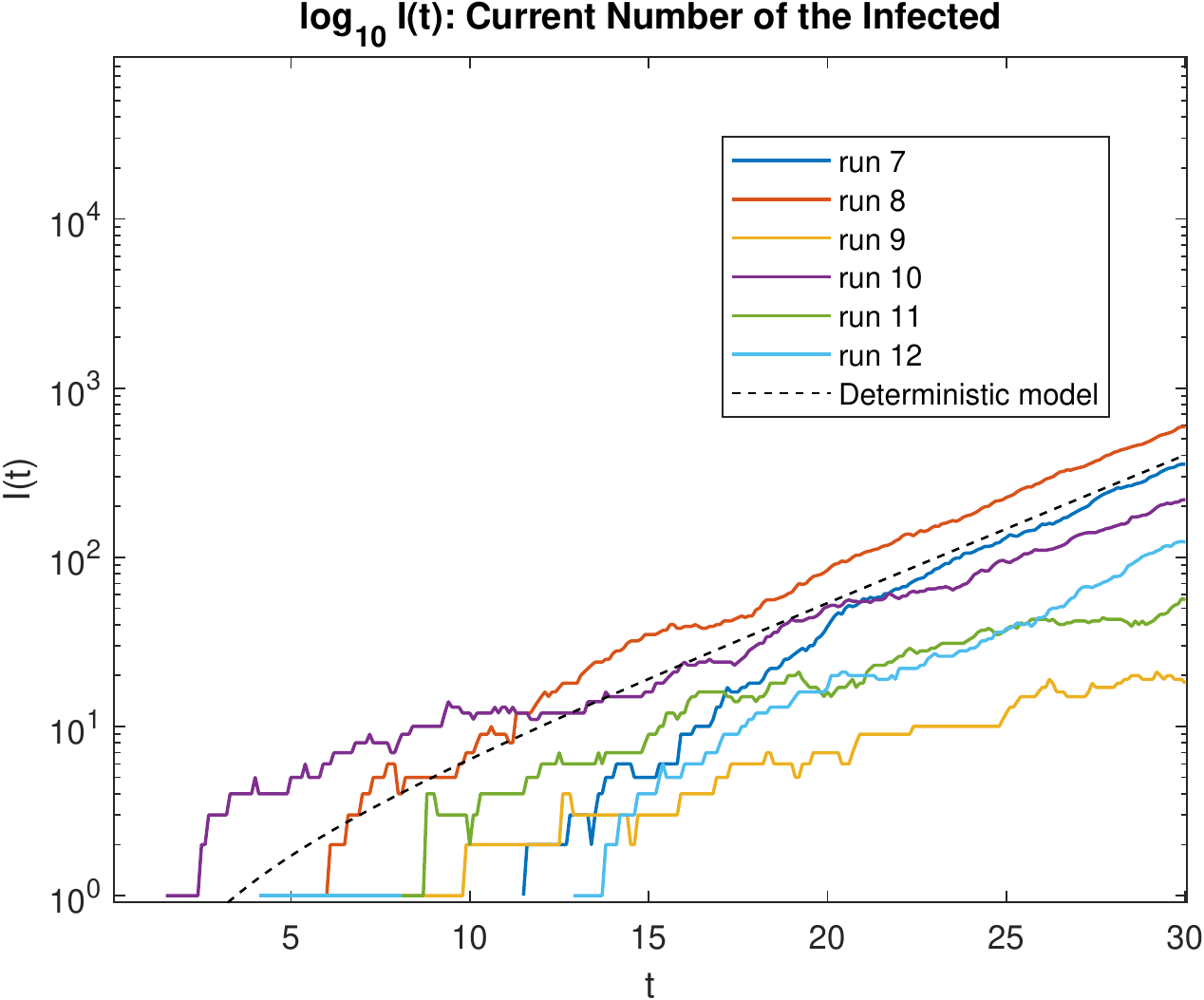}
  \caption{Semi-log plots of the same runs. The baseline of $y$ axis is $10^0=1$.} 
  \label{fig:Semilog_simu_7-12}
  \end{minipage}
\end{figure}
Figure \ref{fig:simu_1-6} shows the first 6 out of a total of 12 simulation runs consecutively done in one execution of our simulation program in a MATLAB script.  Our simulator is based on the \emph{event scheduling approach}\footnote{The time-asynchronous event scheduling approach is a more time efficient and accurate simulation method than a time-synchronous approach. See  e.g., \cite{kobayashi:1978} pp. 230-234, or \cite{kobayashi-mark:2008}, pp. 626-630.}, where events are arrivals of infected persons from the outside (at rate $\nu$ persons/day), occurrences of secondary infections within the community (at rate $\lambda$ infections/day/infectious person), and recovery/deaths of the infected persons (at rate $\mu$ recovery/death/day/infected person).  The parameter values of $\lambda=0.3, \mu=0.1$ and $\nu=0.2$ are assumed.  Figure \ref{fig:Semilog_simu_1-6} plots the same set of curves in a semilog scale. The exponential growth curves are shown as straight-lines; the initial part of the simulation is more clearly shown in the semilog scale.  
Figures \ref{fig:simu_7-12} and \ref{fig:Semilog_simu_7-12} are the plots of the remaining six runs.

There are at least two questions concerning these simulation runs.
\begin{enumerate}
\item Why are the variations among different simulation runs so large?
\item There are more runs whose plots are below the deterministic model curve. Does the deterministic model tend to overestimate the size of the infected population?  If so, why?
\end{enumerate}

The analysis of the BDI process in the following sections should be able to answer these questions.  

\subsection{Formulation for the time-dependent solution for the stochastic model} \label{subsec-Model-Formulation}
Let $I(t)$ represent the number of infected persons at time $t$, and $P_n(t)$ be the time-dependent (or transient) probability mass function (PMF)\footnote{We could use perhaps a more common term the \emph{probability distribution function} but PMF is more explicit that we are dealing with a discrete distribution, not a continuous distribution.} of the process $I(t)$, i.e., 
\begin{align}
P_n(t)=\mbox{Pr}[I(t)=n], ~~n=0, 1, 2, \cdots, ~~\mbox{and}~~t\geq 0.
\label{eq:P_n}
\end{align}

We assume that each infected person is infectious, and infects susceptible persons at rate $\lambda$ [persons/unit-time/person], where the time unit can be arbitrary, e.g., a second, an hour, a day, etc.  Let us assume that an infected person recovers, gets removed or dies at rate $\mu$/[unit-time]. Thus, $\mu^{-1}$~[unit-time] is the \emph{mean infectious period}. The ratio $\lambda/\mu$ is equal to the \emph{basic reproduction number}, i.e., the mean number of infections caused by an infected person during the infectious period. 

We can formulate an infectious disease as a birth-and-death (BD) process, by defining the birth and death rates both of which are simple linear functions of the state $n$ of the process $I(t)$:
\begin{align} 
\lambda_n&=n\lambda + \nu,\nonumber\\
\mu_n&=n\mu,  \label{state-dependent-BD}
\end{align}
A few remarks are in order.  This particular state-dependent BD process is also known as the ``birth-death-immigration (BDI)" process~(see e.g.,\cite{kelly:1979}, p. 14 for the steady state distribution), in which the \textbf{parameters $\lambda$, $\mu$ and $\nu$} represent the \emph{birth} (i.e., secondary infection), \emph{death} (i.e., recovery or death) and \emph{immigration} (i.e., arrival of an infected individual from outside) \emph{rates}, respectively.

A few remarks are in order.
\begin{enumerate}
\item In actuality, ``recovery," ``removal" and ``death" are distinctly different matters. In analyzing the infection process, however,  these three sources of loss from the susceptible or infected population, are mathematically equivalent in the sense they will not contribute to the infection process in the future.  We assume here that those who have recovered from the disease have acquired immunity and will not be susceptible nor infectious.

The assumption that each infected individual recovers (or is removed or dies) at rate $\mu$ is equivalent to assuming that the duration $S$ that each sick person remains infectious is exponentially distributed with mean $1/\mu$, that is:
\begin{align}
\mbox{Pr}[S\leq s]=1-e^{-\mu s},~~s\geq 0.\label{expo-dist}
\end{align}

\item It can be shown mathematically that many of our results to be obtained in terms of the probability mass function (PMF) of $I(t)$, and other related quantities are {\it insensitive} to the actual distribution of $S$ .  All that matters is that we set $\mu=\overline{S}^{-1}$, where $\overline{S}=\Ex[S]$.  

\item In the present paper, we assume that the population is homogeneous, and the susceptible population size remains sufficiently large, thus mathematically treated as ``infinite."  Furthermore, the parameters $\lambda, \mu$ and $\nu$ are assumed to be constant.  Many of our results can be extended to the case of multiple types of populations (e.g., clustering of infections): a model with the susceptible population decreases as some of them get infected; and the case where the model parameters' values change (e.g., the situation where the infectious rate $\lambda$ may change as people's behavior changes), and these generalized models will be discussed in subsequent reports.
\end{enumerate}

We can show (see e.g.,\cite{kobayashi-mark-turin:2012}, pp.407-408) that the PMF (\ref{eq:P_n}) of this BD process should satisfy the following set of linear differential-difference equations, a.k.a. {\it Kolmogorov's forward equation}:
\begin{align}
\frac{dP_0(t)}{dt}&=-\nu P_0(t)+\mu P_1(t)\nonumber\\
\frac{dP_n(t)}{dt}&=\left((n-1)\lambda + \nu\right)P_{n-1}(t)-\left(n(\lambda+\mu)+\nu\right)P_n(t)+(n+1)\mu P_{n+1}(t),~~ n=1, 2, 3 \cdots \label{diff_equation_for_I(t)}
\end{align}
with the initial condition
\begin{align}
I(0)=I_0, ~\mbox{i.e.}, P_n(0)=\delta_{n, I_0},~~n=0, 1, 2, \cdots,
\end{align}
where $\delta_{m,n}$ is Kronecker's delta.  

We transform the above set of infinitely many equations (\ref{diff_equation_for_I(t)}) into a single equation by using  the \textbf{\it probability generating function}~(PGF) (see e.g.,\cite{kobayashi-mark-turin:2012}, p. 402) defined by
\begin{align}
G(z,t)=\Ex[z^{I(t)}]=\sum_{n=0}^\infty z^nP_n(t). \label{def-PGF}
\end{align}
Multiply the set of equations (\ref{diff_equation_for_I(t)}) by $z^n$ and sum them from $n=0$ to $\infty$, obtaining the following partial differential equation:
\begin{equation}
\fbox{
\begin{minipage}{8.5cm}
\[
   \frac{\partial G(z,t)}{\partial t}
=(z-1)\left[(\lambda z -\mu)\frac{\partial G(z,t)}{\partial z}+\nu G(z,t)\right] ,
   \]
\end{minipage}
} \label{PDE-for-PGF}
\end{equation}
with the boundary condition
\begin{align}
G(z,0)=z^{I_0}. \label{initial}
\end{align}

\subsection{Stochastic means of the infected process \emph{I(t)} and related processes} \label{subsec-Expectation}

Although the process $I(t)$ is the main focus of our analysis,  it will be worthwhile to  introduce  related processes and our assumptions.
\begin{definition} \label{def-A(t)-etc}
\noindent
\begin{enumerate}

\item The process $A(t)$ is the cumulative count of external arrivals of infectious individuals from the outside. We assume that $A(t)$ is a Poisson process with rate $\nu$ [persons/unit time],

\item The process $B(t)$  is the cumulative count of internally infected individuals.  We assume that the birth of such persons occurs at the rate of $\lambda$ [persons/unit time/infectious person].  

\item The process $R(t)$ is the cumulative count of recovered/removed or dead individuals.  We assume that the departure of such persons occurs at the rate of $\mu$ [persons/unit time/infected person].  Note that all infected  persons are infectious persons until their recovery/removal/death.

\item The process $I(t)$ is the present number of infected persons, i.e.,
\begin{align}
I(t)=I(0)+A(t)+B(t)-R(t).  \label{I-A-B-D}
\end{align}
$\Box$
\end{enumerate}
\end{definition}

The expectation and variance of the above processes will be of our interest, which we denote by
\begin{align}
\oA(t)&=\Ex[A(t)],~~\oB(t)=\Ex[B(t)],~~\oR(t)=\Ex[R(t)],~~\mbox{and}~~\oI(t)=\Ex[I(t)], \label{Mean-A-etc}\\
\sigma^2_A(t)&=E[(A(t)-\oA(t))^2],~~ \mbox{etc.}  \label{Variance-A-etc}
\end{align}

Before we discuss how to find the PGF $G(z,t)$ from the PDE (\ref{PDE-for-PGF}), let us derive first an ordinary differential equation for $\oI(t)$. 
By dividing both sides of (\ref{PDE-for-PGF}) by $(z-1)G(z,t)$, we will have 
\begin{align}
(z-1)^{-1}\frac{\partial\ln G(z,t)}{\partial t}&=
(\lambda z -\mu)\frac{\partial \ln G(z,t)}{\partial z}+\nu.  \label{partial_lnG}
\end{align}
By setting $z=1$, we find\footnote{Alternatively, we can obtain this differential equation  directly, by multiplying each equation in
(\ref{diff_equation_for_I(t)}) by $n$ and summing them up from $n=0$ to infinity.}
\begin{align}
\frac{d\oI(t)}{dt}=  (\lambda-\mu)\oI(t) + \nu. \label{diff_for_I}
\end{align}
where, on the LHS\footnote{The abbreviations LHS and RHS mean the left-hand side and right-hand side, respectively.}, we first set $z=1$ (which corresponds to differentiation at $z=1$), and use L'H\^{o}pital's rule, obtaining \footnote{Here we use an important property of PGF, i.e, $\frac{\partial G(z,t)}{\partial t}=\Ex[I(t)z^{I(t)-1}], $ and by setting $z=1$, the RHS becomes $\Ex[I(t)]$. We changed the order of differentiation w.r.r. to $z$ and $t$, which can be justified because the function $G(z,t)$ is an analytic function, i.e., it is continuous and differentiable everywhere.} $\Ex[I(t)]=\oI(t)$. 
The ordinary differential equation (\ref{diff_for_I}) can be solved, yielding

\begin{align}
\oI(t)&=I_0e^{a t}+\frac{\nu}{a}\left( e^{a t}-1 \right),~~t\geq 0,~~\mbox{where}~~I_0=I(0)~~ \mbox{and}~~a=\lambda-\mu.
  \label{mean-I(t)}
\end{align}
If the model parameters are set to new values, say, to $\lambda', \mu'$ and $\nu'$ at some point $t=t_1\geq 0$, then the solution $\oI(t)$ for $t\geq t_1$ is given by
\begin{align}
\oI(t)&=I_1 e^{a'(t-t_1)}+ \frac{\nu'}{a'}\left(e^{a'(t-t_1)}-1\right),~~~~t\geq t_1,~~~\mbox{where}~~I_1=\oI(t_1),~~a'=\lambda'-\mu' .  \label{I(t)-after-t_1}
\end{align}
It should be clear that  $\oI(t)$ diverges to infinity, if $a>0$ in (\ref{mean-I(t)}) and converges to $\nu/|a|$ in the limit $t\to\infty$ if $a<0$.  Similarly, in  (\ref{I(t)-after-t_1}), the process converges to  $\nu'/|a'|$, if $a'<0$.

If $a=0$, then\footnote{The second term becomes 0/0, so we apply L'H\^{o}pital's rule.}
\begin{align}
\oI(t)=I_0+\nu t=I_0+A(t)~~t\geq 0,~~\mbox{when}~~a=0. 
\end{align}

In Figure \ref{fig:deterministic} we show the case where $\lambda=0.3, \mu=0.1$ and $\nu=0.2$ and at $t_1=30$, a new parameter $\lambda'=0.06$ is set, whereas the original values of $\mu$ and $\nu$ are retained.
\begin{figure}
\centering
\includegraphics[scale=0.6]{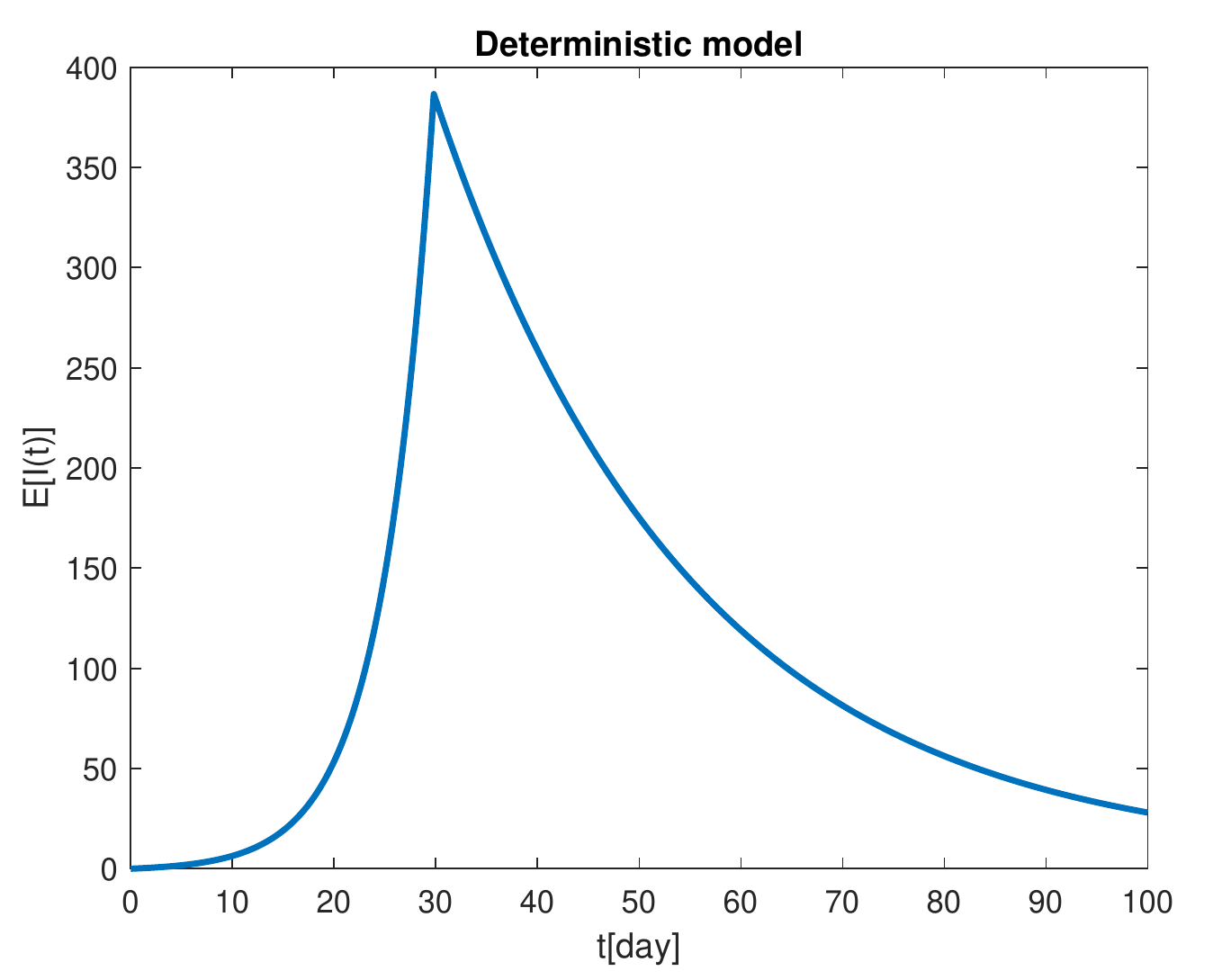}
\caption{$\oI(t)$ when $\lambda=0.3, \mu=0.1, \nu=0.2$ for $0\leq t\leq t_1=30$ at which point $\lambda$ is changed to $\lambda'=0.06$.}
\label{fig:deterministic}
\end{figure}

The mean values of other processes can be easily found.  $A(t)$ is a Poisson process with rate $\nu$, which implies 
\begin{align}
\oA(t)=\nu t,~~\mbox{for all}~~t\geq 0. \label{mean-A(t)}
\end{align}
Since each person in the infected population $I(t)$ infects others at the rate of $\lambda$ persons/unit time, the differential of $\oB(t)$ is given by
\begin{align}
\frac{d\oB(t)}{dt}=\lambda \oI(t). \label{Eq-for-B(t)}
\end{align}

From this and (\ref{mean-I(t)}) we obtain
\begin{align}
\oB(t)&=\lambda\int_0^t \oI(u)\,du =\frac{\lambda I_0 (e^{at}-1)}{a}+\frac{\lambda \nu (e^{at}-1)}{a^2}-\frac{\lambda\nu t}{a} \nonumber\\
&=\frac{\lambda}{a}(\oI(t)-I_0-\oA(t)),~~0\leq t \leq t_1. \label{mean-B(t)}
\end{align}

For $t\geq t_1$, we find
\begin{align}
B(t)&=B(t_1)+\frac{\lambda'I_1(e^{a'(t-t_1)}-1)}{a'}-\frac{\lambda'\nu'(t-t_1)}{a}
+\frac{\lambda'\nu'(e^{a(t-t_1)}-1)}{a^2}\nonumber\\
&=\frac{\lambda'}{a'}\left(I_1e^{a'(t-t_1)}+\frac{\nu'}{a'}( e^{a'(t-t_1)}-1)-I_0-\nu't\right)\nonumber\\
&=\frac{\lambda'}{a'}(I(t)-I_0-\oA'(t)),~~\mbox{where}~~\oA'(t)=\nu't,  \label{B(t)-after-t_1}
\end{align}
which takes the same form as that for $t\leq t_1$. It should be worthwhile to note that both (\ref{mean-B(t)}) and (\ref{B(t)-after-t_1}) could be derived from the mean value of the identity equation (\ref{I-A-B-D}) together with the relation $\oR(t)=\frac{\mu}{\lambda}\oB(t)$, which is evident as shown below. 

The recovery process $\oR(t)$ should satisfy the following differential equation, similar to (\ref{Eq-for-B(t)}):
\begin{align}
\frac{d\oR(t)}{dt}&=\mu \oI(t). \label{Eq-for-R(t)}
\end{align}
Thus, it readily follows:
\begin{align}
\oR(t)&=\left\{\begin{array}{ll}
\mu(\oI(t)-I_0-\oA(t))/a,~~~& 0\leq t\leq t_1\\
\mu'(\oI(t)-I_0-\oA'(t))/a',~~~& t\geq t_1. \end{array}\right.
  \label{mean-R(t)}
\end{align}

The above expressions for $\oI(t)$ and other processes can be viewed as our \textbf{\emph{deterministic (or non-probabilistic)  model}} for the dynamics of the BDI process.  From these simple expressions, we can extract a few important characteristics concerning the mean values of $I(t)$ and others. 
\begin{enumerate}
\item  $a=\lambda-\mu$ determines the exponential growth or decay rate of $\oB(t), \oR(t)$ as well as $\oI(t)$.

\item  $\nu$ is merely a linear scaling factor for $\oI(t)$ and other processes, and so is $I_0$, the initial number of the infected.\footnote{Eq.(\ref{mean-I(t)}) can be rewritten as  $\oI(t)=\left(I_0+\frac{\nu}{a}\right) e^{a t}-\frac{\nu}{a}$. So $\frac{\nu}{a} + I_0 $ is the multiplying coefficient of the exponential term $e^{a t}$.}

\item  If $a>0$, then $\oI(t)$ grows exponentially without bound; if $a<0$, it decays exponentially towards $\nu/|a|$. If $a=0$,  $\oI(t)=I_0+\nu t$, i.e., $\oI(t)$ grows linearly.

\item  The ratio of the infection rate (or reproduction rate)
$\lambda$ to the recovery or removal rate $\mu$
\begin{align}
R_0 = \frac{\lambda}{\mu} \label{reproduction_no}
\end{align}
is  called the \textbf{\emph{basic reproduction number}} in epidemiology (see  e.g., \cite{martcheva:2010}, p. 21). The term was originally defined in the context of a deterministic model called the \textbf{SIS}(\emph{Susceptible-Infected-Susceptible}) epidemic model.   It is the average number of persons whom an infectious person infects 
before his/her recovery, removal, or death. 
The reproduction number determines whether the infection will grow exponentially, die out, or remain constant, depending on whether $R_0>1$, $R_0<1$, or $R_0=1$, respectively.  
The exponential parameter $a$ can be expressed in terms of $R_0$ and $\mu$:
\begin{align}
a=\lambda -\mu =(R_0-1)\mu.  ~~\label{alpha-R}
\end{align}

\item  The amount of time $T$ that takes for $\oI(t)$ or other related quantities to double, and thes exponential growth parameter $a$ are related by
\begin{align}
e^{a T}=2, ~~\mbox{or equivalently}~~a T \approx 0.693.\label{doubling-days}
\end{align}
Note that both the integral and derivatives of the exponential function $e^{a t}$ are also $\propto e^{a t}$.  Thus, the above formula for $T$  equally applies, when cumulative numbers or incremental numbers are to be counted for a given $a$ via an observed $T$.

\item Unless we can expect to increase the value of $\mu$ by improving the medical service or producing an effective vaccine to immunize the susceptible population, the only options we have for controlling an infectious disease is to \textbf{\emph{increase $\mu$}} by removing as many infectious individuals away from susceptible population as possible, and/or to \textbf{\emph{decrease $\lambda$}} by increasing the so-called \textbf{\emph{social distances}} between the susceptible and the infectious. We will provide an illustrative example in Part II \cite{kobayashi:2020b}.
\end{enumerate}

\horizline
\noindent\\
\small{\textbf{Example 1}\label{Example_1-deterministic}
Consider the following community:~~
The external arrival rate of infected individuals $\nu=0.2$ [persons/day], i.e. one such such incidence every 5 days on average. The average number of days required for an infected person to recover, be removed or die, is $\mu^{-1}=10$ [days], i.e., the recovery rate $\mu=0.1$ [per day] for each infected person.  Suppose that the average number of secondary infections caused by an infected individual is estimated as $R_0=\lambda/\mu=3$ for each infectious person.  Then this value and $\mu$ provide an estimate of the infectious rate $\lambda=0.3$ [persons/day/infected person].  Thus, the exponential growth parameter is estimated as $a=\lambda-\mu=0.2>0$.  One can check the validity of the model and/or consistency among the three estimated parameters, by computing a second estimate of the $a$ from the formula (\ref{doubling-days}).  In Part II, we will discuss various ways of estimating the model parameters.

Once we have obtained reliable model parameters, we can predict the expected value of the infected process $\oI(t)$, assuming that there was no infected person at $t=0$, i.e., $I(0)=I_0=0$:   
\begin{align}
\oI(t)=\frac{\nu}{a}\left[e^{a t}- 1~\right]=e^{0.2 t}- 1.\label{ex-determinisitc-curve}
\end{align}
Thus, this curve predicts the expected number of the infected population, excluding those who have recovered, been removed or have died.
\begin{align}
\oI(0)&=0, ~~\oI(5)=1.7,~~\oI(10)=15.0,~~\oI(15)=19.1,~~\oI(20)=53.6,~~\oI(25)=147.4,\nonumber\\
\oI(30)&=402.5,~~\oI(35)=1,095.5,~~\oI(40)=2,980.0,~~\oI(45)=8,102.1,~~\oI(50)=22,025.5.  \nonumber
\end{align} $\Box$

\subsection{Steady-state distribution of the \emph{I(t)}}\label{subsec-steady-state}
So far we have discussed only the mean values of the random process $I(t)$ and other processes.
Before we find the probability mass functions $P_n(t), n=0,1, 2, \cdots$ for any $t$, we obtain in this section the steady-state distribution $\lim_{t\to\infty}P_n(t)=\pi_n$, if it exists.  We know already that when $a>0$, such distribution does not exist. 

Thus, the steady state distribution can possibly exist, only when $a\leq 0.$ In order to find it,  we set the LHS of the PDE (\ref{PDE-for-PGF}) equal to zero, obtaining the following ordinary differential equation for the PGF $G(z, \infty)$:
\begin{align}
(\lambda z-\mu)\frac{d G(z,\infty)}{dz}+\nu G(z, \infty)=0,
\end{align}
which readily leads to
\begin{align}
\frac{d \ln G(z,\infty)}{dz}=-\frac{\nu}{\lambda z-\mu}.
\end{align}
Integrating the above and using the boundary condition $G(1,t)=1$ for any $t$, we find
\begin{align}
G(z, \infty)=\left(\frac{1-\frac{\lambda}{\mu}}{1-\frac{\lambda}{\mu} z}\right)^r,
~~\mbox{where}~~r=\frac{\nu}{\lambda}. \label{PGF-steady}
\end{align}
This PGF reminds us of the \textbf{\emph{negative binomial distribution}} (NBD).  This distribution was originally introduced to express the probability of the \emph{number of failures $n$} needed to achieve $r$ successes in a sequnce of Bernoulli trials, when the probability of failure  is $q$. 

\begin{definition}[Negative binomial distribution (NBD)] \label{def-negative-binomial}
Negative binomial distribution NB$(r,q)$ is defined by
\begin{align}
P^{\scriptscriptstyle NB}_n={n+r-1\choose n} (1-q)^r q^n                                                                                                                                                                                                                                                                                                                                                                                                                                                                                                                                                                                                                                                                                                                                                                                                                                                                                                                                                                                                                                                                                                                                                                                                                                                                                                                                                                                                                                                 =\frac{\Gamma(n+r)}{n!\Gamma(r)}(1-q)^r q^n                                                                                                                                                                                                                                                                                                                                                                                                                                                                                                                                                                                                                                                                                                                                                                                                                                                                                                                                                                                                                                                                                                                                                                                                                                                                                                                                                                                                                                                , ~~n=0, 1, 2, \cdots, \label{def-NBD}
\end{align}
where the parameter $r$ is a positive real number. $\Box$
\end{definition}
When $r$ is a positive integer, and $0<q<1$ the above reduces to the classical definition of the (shifted)\footnote{Some authors define the negative binomial distribution as the distribution of the \emph{number of trials}, instead of the \emph{number of failures}, needed to achieve $r$ successes. Under this definition, $n=r, r+1, r+2, \cdots$. The probability distribution (\ref{def-NBD}) is then referred to as the \textbf{\emph{shifted negative binomial distribution}} (see e.g.,\cite{kobayashi-mark-turin:2012}, pp. 59-62).}  negative binomial distribution, sometimes called the Pascal distribution,  associated with Bernoulli trials.
The Gamma function $\Gamma(x)$ is defined for a positive real number $x$ by\footnote{This definition can be extended for a complex number $z$, with $\Re(z)>0$.} 
\begin{align}
\Gamma(x)=\int_0^\infty y^{x-1}e^{-y}\,dy,~~x>0.
\end{align}
The \emph{probability generating function} of NB$(r, q)$ of (\ref{def-NBD}) is given by
\begin{align}
G(z)&=\sum_{n=0}^\infty {n+r-1\choose n} (1-q)^r q^n                                                                                                                                                                                                                                                                                                                                                                                                                                                                                                                                                                                                                                                                                                                                                                                                                                                                                                                                                                                                                                                                                                                                                                                                                                                                                                                                                                                                                                                 z^n=\frac{1-q}{1-qz}, ~~|z|<q^{-1}.   \label{PGF-NBD}                                            
\end{align}  
The mean and variance of a RV (random variable) $X$ possessing this distribution can be readily found:
\begin{align}
\Ex[X]=\frac{qr}{1-q},~~\mbox{and}~~\mbox{Var}[X]=\frac{qr}{(1-q)^2}.  \label{mean-var-NBD}
\end{align}    
From (\ref{PGF-steady}) and the above formula,  we readily find the steady state distribution of $I(t)$ for $a<0$ is given as follows \cite{kelly:1979},p.14.
\begin{equation}
\fbox{
\begin{minipage}{7.5cm}
\[
\pi_n=\lim_{t\to\infty}P_n(t)={n+r -1\choose n}\left(\frac{\lambda}{\mu}\right)^n\left(1-\frac{\lambda}{\mu}\right)^r.
\]
\end{minipage}
} \label{steady-state-pi}
\end{equation}

\section{Time-Dependent Probability Distribution of the Infected Process $I(t)$} 
The partial difference equation (PDE) (\ref{PDE-for-PGF}) for the PGF $G(z,t)$ is  a linear PDE and is sometimes referred to as \textbf{\emph{planar differential equation}} (see \cite{gross-harris:1985}, \cite{kobayashi-mark:2008}, pp. 600-605). This type of PDE can be solved by using \textbf{\emph{Lagrange's method}} with its \emph{auxiliary differential equations}, which is discussed in Appendix A. \footnote{It will be worth noting that Ren and Kobayashi \cite{ren-kobayashi:1995}, \cite{kobayashi-ren:1992} discuss this type of PDE in the analysis of multiple on-off sources in traffic characterization of a data network.}

\subsection{When the system is initially empty, i.e., $I(0)=0$} \label{subsec-BDI-is-BN-distributed}

When the system is initially empty, i.e., $I(0)=0$, its PGF can  be found by solving the PDE given by (see (\ref{PGF-empty}) of Appendix A)
\begin{align}
G(z,t)&=\left(\frac{a}{\lambda e^{a t}-\mu -\lambda(e^{a t}-\mu)z}
\right)^r, ~~\mbox{where}~~a=\lambda-\mu,~~\mbox{and}~~r=\frac{\nu}{\lambda}.  \label{PGF-empty-1}
\end{align}
If we define a function $\beta(t)$ 
\begin{align}
\beta(t)=\frac{\lambda(e^{a t}-1)}{\lambda e^{a t}-\mu}, \label{def-beta(t)}
\end{align}
we can write (\ref{PGF-empty-1})  compactly as
\begin{equation}
 \fbox{
\begin{minipage}{5.5cm}
\[
G(z,t)=\left(\frac{1-\beta(t)}{1-\beta(t)z}\right)^r,    
\] 
\end{minipage}
}\label{PGF-solution}
\end{equation}
which is the PGF of the generalized negative binomial distribution NB($r, \beta(t)$), defined in Definition \ref{def-negative-binomial}.  
Thus, we find that the PMF of the BDI process $I(t)$ is given by
\begin{align}
P_n(t)=\mbox{Pr}[I(t)=n]={n+r-1\choose n}(1-\beta(t))^r \beta(t)^n,~~n=0, 1, 2, \cdots. \label{time-dependent-sol}
\end{align}

When $a=0$, the above expression can be simplified.  Noting
\begin{align}
\lim_{a\to 0}\frac{e^{a t}-1}{a} = t,
\end{align}
we find
\begin{align}
\beta(t)=\frac{\mu t}{\mu t+1},~~\mbox{and}~~1-\beta(t)=\frac{1}{\mu t+1}, 
\end{align}
which lead to
\begin{align}
G(z,t)=\left( \frac{1}{1+\mu t-\mu t z}\right)^r, 
\end{align}
and
\begin{align}
P_n(t)={n+r-1\choose n}\left(\frac{1}{\mu t+1}\right)^r \left(\frac{\mu t}{\mu t+1}\right)^n.
\end{align}

The mean and variance can be computed from the distribution function obtained above.  But we already know the expression for the mean $\oI_{\scriptscriptstyle BDI}(t)$ from (\ref{mean-I(t)}), which can be also found from the formula of the NB($r,q$, i.e.,
\begin{align}
\Ex[I_{\scriptscriptstyle BDI}(t)]=\frac{r\beta(t)}{1-\beta(t)}=\frac{\nu}{\alpha}(e^{at}-1),
\label{Expectation-BDI}
\end{align}
which certainly agrees with (\ref{mean-I(t)}).  Similarly, the variance can be found from the formula of NB($r, q$) as
\begin{equation}
\fbox{
\begin{minipage}{8.5cm}
\[
   \sigma_{\scriptscriptstyle BDI}^2(t)=\frac{r\beta(t)}{(1-\beta(t))^2}=\frac{\nu(\lambda e^{at}-\mu)(e^{at}-1)}{a^2}.
   \]
\end{minipage}
} \label{sigma-square-BDI}
\end{equation}
The reason why this simple innocent-looking equation is put in a box is because this expression for \textbf{\emph{the variance is perhaps the most important feature}} behind the erratic behavior we observe in the COVID-19, or an infectious disease in general, as we analyze below and demonstrate by presenting simulation results in Part II.

\subsection{The process \emph{I(t)} with an arbitrary initial condition} \label{subsec-Arbitrary-initial}

Up to now, we have focused on the case when the system is initially empty. In this section, we wish to study a general case where $I(0)=I_0$ is an arbitrary non-negative integer.  The PGF solution for this general case is, in fact, given in (\ref{PGF-general-BDI}) of Appendix A.  In this section we will add to the PGF and PMF of this general case a subscript or superscript $\scriptscriptstyle BDI: I_0$ wherever we need to distinguish similar functions of other processes.
So we can write
\begin{align}
G_{\scriptscriptstyle BDI: I_0}(z,t) &= G_{\scriptscriptstyle BDI: I_0=0}(z,t) ~~
G_{\scriptscriptstyle BD:I_0}(z,t),\label{BDI-PGF-as-product}
\end{align}
where 
\begin{align} 
G_{\scriptscriptstyle BDI:I_0=0}(z,t)&=\left(\frac{a}{\lambda e^{a t} -\mu -\lambda(e^{a t}-1)z}\right)^r\label{BDI-empty}\\
G_{\scriptscriptstyle BD:I_0}(z,t)&=~\left(\frac{\mu e^{a t}-\mu +(\lambda- \mu e^{a t})z}{\lambda e^{a t}-\mu-\lambda (e^{a t}-1)z}\right)^{I_0},\label{BD-I_0}
\end{align}
where the first given by (\ref{BDI-empty}) is the same as (\ref{PGF-empty-1}) discussed in a preceding section, whereas the second one (\ref{BD-I_0}) is the PGF of \textbf{\emph{the birth-and-death}} process,  
obtained by setting $\nu=0$ in the PGF (\ref{PGF-general-BDI}).
Thus, we see that the BDI:$I_0$ process can be expressed as a sum of two statistically independent processes:
\begin{align}
I_{\scriptscriptstyle BDI:I_0}(t)=I_{\scriptscriptstyle BDI:I_0=0}(t)+I_{\scriptscriptstyle BD: I_0}(z,t).
\end{align}

It can be further shown that the the birth-and-death process can be decomposed into two processes: the \textbf{\emph{the pure-birth process}} and a generalized binomial distributed process, denoted $I_{\scriptscriptstyle GB:p(t)}(t)$, where $p(t)$ is given by (\ref{def-p(t)}). 
Thus, 
\begin{align}
I_{\scriptscriptstyle BD:I_0}(t)=I_{\scriptscriptstyle NB(I_0,\beta(t))}(t)+I_{\scriptscriptstyle GB:I_0}(t),
\end{align}
and the PGFs of these processes can be expressed as
\begin{align}
G_{\scriptscriptstyle BD:I_0}(z,t)&=G_{\scriptscriptstyle NB(I_0,\beta(t))}(z, t)~ G_{\scriptscriptstyle GB:I_0}(z, t)\label{PGF-BD-NB-PD}
\end{align}
where
\begin{align}
G_{\scriptscriptstyle NB(I_0,\beta(t))}(z,t)&=\left(\frac{1-b(t)}{1-b(t)z}\right)^{I_0}   \label{PGF-NB}  \\
G_{\scriptscriptstyle GB:p(t)}(z,t)&=\left(1-p(t) + p(z) z\right)^{I_0}  \label{PGF-PD}
\end{align}

\section{Important Properties of the  Negative Binomial Distribution (NBD)}

\subsection{Coefficient-of-variation of an NBD with small $r$}

Let us further examine properties of a random variable $X$ of NB($r, q$) defined by its PMF (\ref{def-NBD}) and PGF (\ref{PGF-NBD}).
Its mean and variance are given in (\ref{mean-var-NBD}). 
The mode of the this distribution is given by
\begin{align}
\mbox{mode}_{X_{\scriptscriptstyle NB}}=\left\{\begin{array}{lll}
 & \lfloor\frac{q(r-1)}{1-q}\rfloor,~~&\mbox{if}~~r>1\\
 &  0,~~&\mbox{if}~~ r\leq 1. 
 \end{array}\right.  \label{mode-NBD}
\end{align}

Recall that the parameter $r$ in our problem is the ratio of $\nu$ [person/unit-time] (the rate with which an infected person arrives in the community in question) to $\lambda$ [person/unit-time/person] (the rate with which an infectious person infects susceptible people. As long as tight security measures are enforced at the boundaries of the community,  
$\nu$ is kept small.  Thus, in our problem of practical interest, $r < 1$, or even $r\ll 1$.  

Figure \ref{fig:Pascal} shows the (shifted) negative binomial distribution (NBD) for $p=1-q=0.5$ and for various values of $r=0.5, 1 ,2,4, 8, 16, 32$.
\begin{figure}
\centering
\includegraphics[scale=0.5]{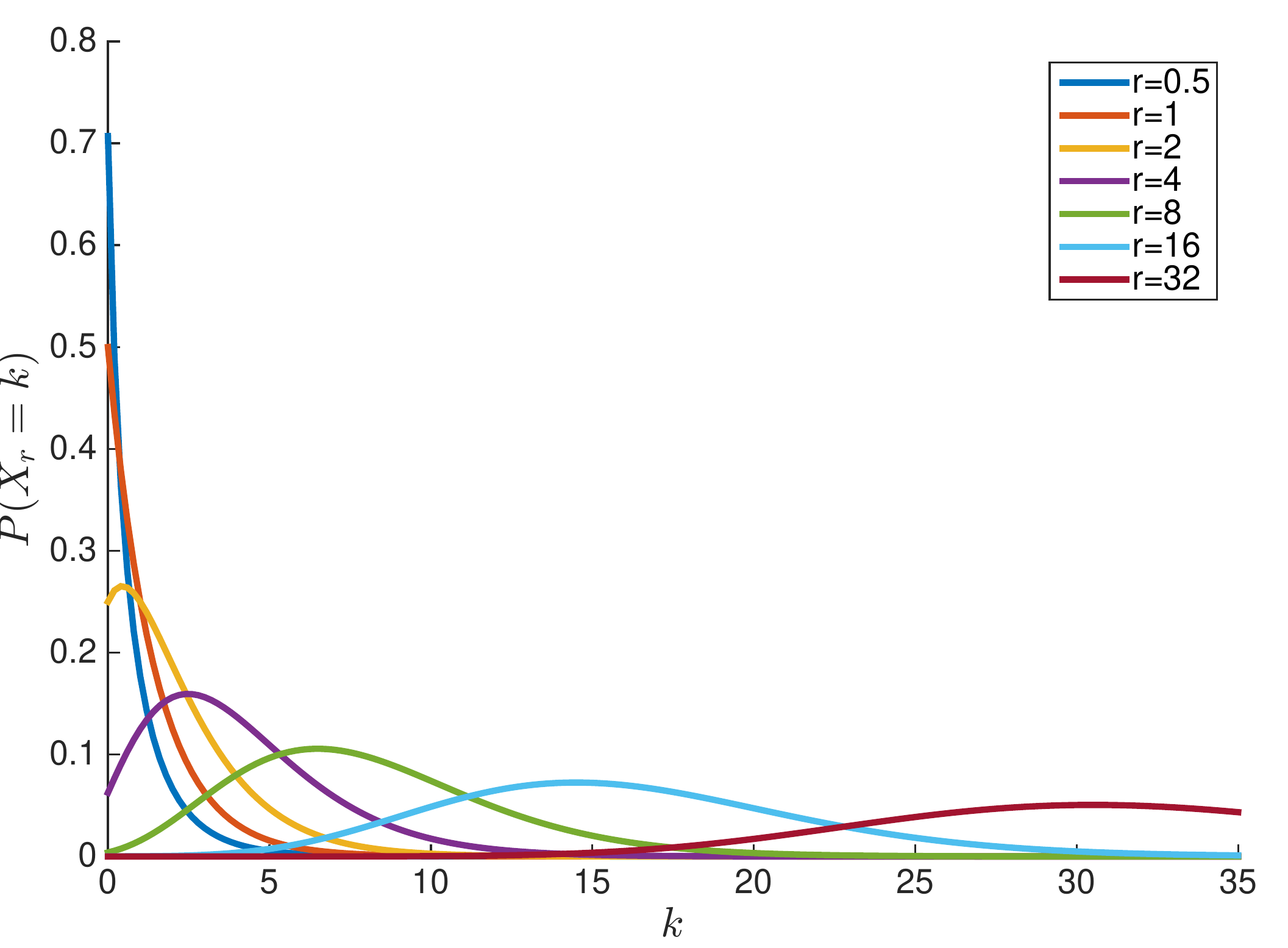}
\caption{\footnotesize The (shifted) negative binomial distribution (\ref{def-NBD}), i.e., the probability distribution of the number of failures needed to achieve $r$ successes in Bernoulli trials with $q(=1-p)=0.5$ and for $r=0.5, 1, 2, 4,8, 16, 32.$:
}
\label{fig:Pascal}
\end{figure}

As we can see in Figure \ref{fig:Pascal}, the PMF with $r=0.5$ falls off rapidly by around $n\approx 5$. For $q=0.5$ and $r=0.5$, we compute the first several values of $P_n$:
\begin{align}
P_0=0.7071,~~ P_1=0.5303,~~ P_2=0.3315,~~ P_3= 0.1934,~~
   P_4=0.1088,~~P_5=0.0598, \cdots, P_{10}=0.00256.
\end{align}

In order to see more clearly how the PMF $P_n$ looks for small $r$, let us rewrite (\ref{def-NBD}) as follows:
\begin{align}
P_n=k(r) \frac{(n-1+r)}{(n-1)}\frac{(n-2+r)}{(n-2)}\cdots \frac{(2+r)}{2}\frac{(1+r)}{1} \frac{q^n}{n}, ~~\mbox{where}~~k(r)=r p^r.
\end{align}
The leading term $k(r)$ does not depend on $n$.  So we can write
\begin{align}
P_n \propto \frac{(n-1+r)}{(n-1)}\frac{(n-2+r)}{(n-2)}\cdots \frac{(2+r)}{2}\frac{(1+r)}{1} \frac{q^n}{n}.
\end{align}
In the limit $r\to 0$, the shape of $P_n$ will become
\begin{align}
P_n\propto  \frac{q^n}{n}.  \label{P_n-for r=0}
\end{align}
Recall the following Maclaurin series expansion 
\begin{align}
-\ln(1-q)=q+\frac{q^2}{2}+\frac{q^3}{3}+\cdots,  \label{Maclaurin}
\end{align}
from which we obtain
\begin{align}
\sum_{n=1}^\infty\frac{-1}{\ln(1-q)}\frac{q^n}{n}=1.  \label{log-dist-identity}
\end{align}
Thus, we find (\ref{P_n-for r=0}) can be written as
\begin{align}
P_n&=\frac{1}{-\ln (1-q)}\frac{q^n}{n},~~n=1, 2, 3, \cdots.\label{Fisher}
\end{align}
This distribution is called the \textbf{\emph{logarithmic distribution}} (a.k.a. \textbf{\emph{logarithmic series distribution}}), which was introduced by R. A. Fisher in 1943 \cite{fisher:1943}. The PGF of (\ref{Fisher}) is given, using (\ref{Maclaurin}) once again, by
\begin{align}
G_{\scriptscriptstyle \log}(z)&=\frac{\ln(1-qz)}{\ln (1-q)}.
 \label{PGF-Fisher}
\end{align}

In our case the probability $q$ is given by $\beta(t)$ of (\ref{def-beta(t)}):
\begin{align}
q=\beta(t)=\frac{\lambda(e^{at}-1)}{\lambda e^{at}-\mu}, ~~\mbox{where}~~a=\lambda-\mu. \label{q=beta(t)}
\end{align}
Thus, for sufficiently large $t$, the parameter $q$ is very close to one, thus the term $q^n$ is a slowly decreasing geometric series.  For small $r$, the coefficient is also slowly decaying.  Thus, for small $r$ and $q$ close to 1, the NBD is a very slowly decaying distribution. having a long tail, similar to \textbf{\emph{Zipf's law}} or the \textbf{\emph{zeta distribution}} with small power exponent $\alpha$ (see e.g., \cite{kobayashi-mark-turin:2012}, pp. 62-64).
 
\noindent\\
{\small{\textbf{Example 2:}} \label{Example-NBD}
Consider the same set of model parameters used in Example 1, i.e., $\lambda=0.3, \mu=0.1$ and $\nu=0.2$.  Then
$r=\frac{\nu}{\lambda}=\frac{2}{3}$, and $a=\lambda-\mu=0.2$.  In Figure \ref{fig:NBD_t=1} through Figure \ref{fig:NBD_t=50}, we show the PMF $P_n(t)$ (\ref{time-dependent-sol}) of the infected population $I(t)$ at $t=1, 5, 10, 20, 30$ and $50$. In all these cases, $I_0=0$, i.e., there are no infected individuals at $t=0$.  

One useful method to see clearly the tail end of the distribution function is to plot the \textbf{\emph{log-survival function}}, $\log_{10}(1-F_n)$, where $F_n$ is the cumulative distribution function (CDF), i.e., $F_n=\sum_{i=0}^n P_n(t),$ and its complement $1-F_n$ is called the survival (or survivor) function.\footnote{The term ``survival (or survivor) function"  is often used in reliability theory. If a continuous variable $X$ represent the ``life" of a human, or a product (e.g., electric bulb), with its probability density function (PDF) $f_X(x)$ and cumulative distribution function (CDF) $F_X(x)=\int_0^t f_X(s)\,ds$ , the survival function is defined by $1-F_X(x)$.  By plotting $\log(1-F_X(x)$ vs. $x$, we find more clearly, how the remaining life will behave than a regular plot of $f_X(x)$ or $F_X(x)$.  See e.g., \cite{kobayashi-mark-turin:2012} pp. 144-146.}  In Figure \ref{fig:log_surv_t=10} and Figure \ref{fig:log_surv_t=30}, we show the log-survivor functions of the distribution of $I(t)$ at $t=10$ and $t=30$. We confirm our earlier observation that the NBD behaves similar to a geometric distribution, which is a straight line in the log-survivor plot.
\begin{figure}[th]
  \begin{minipage}[b]{0.45\linewidth}
  \centering
  \includegraphics[width=\textwidth]{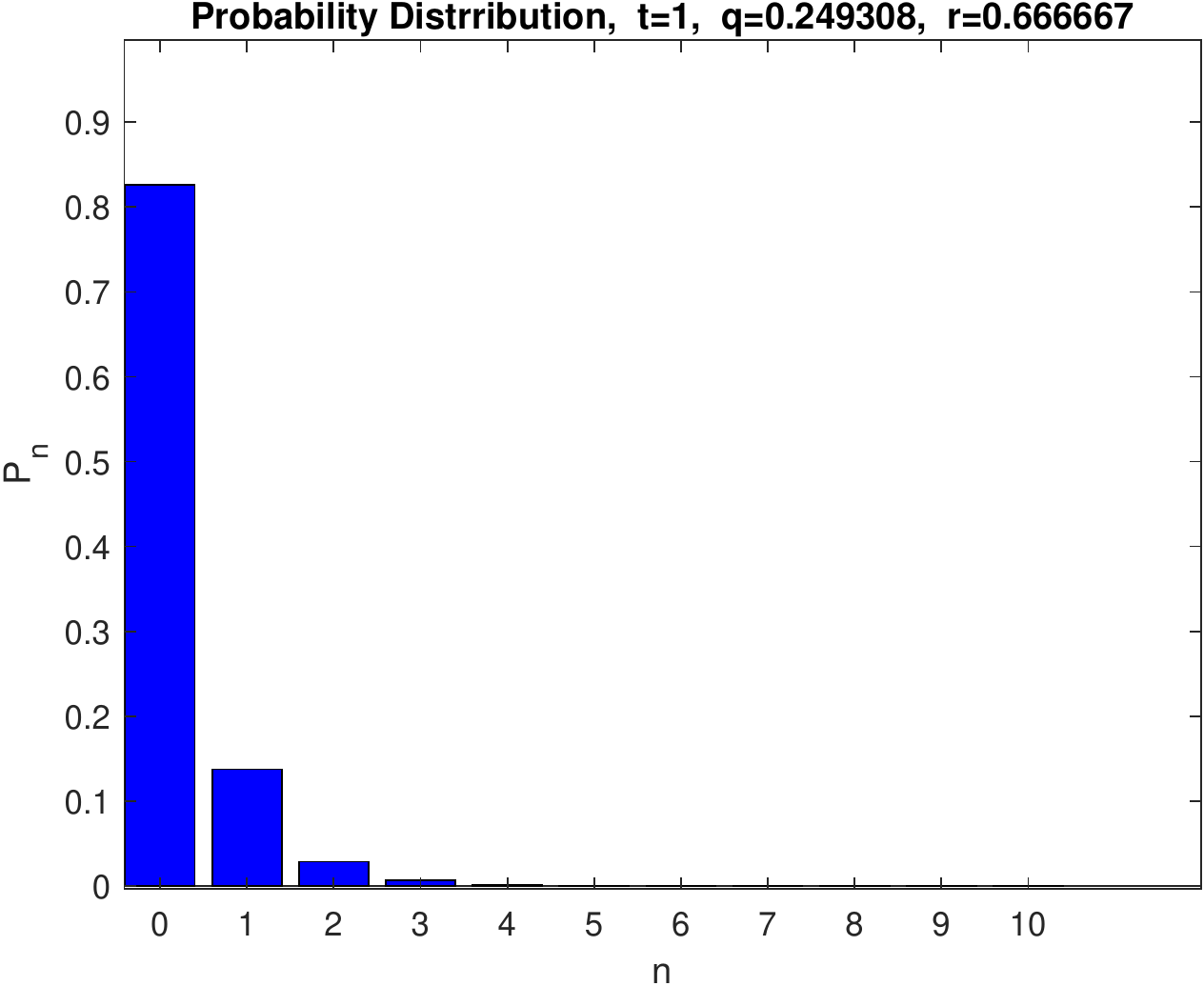}
  \caption{$P_n(t)$ at $t=1$.}
  \label{fig:NBD_t=1}
  \end{minipage}
  \hspace{0.5cm}
  \begin{minipage}[b]{0.45\linewidth}
  \centering
  \includegraphics[width=\textwidth]{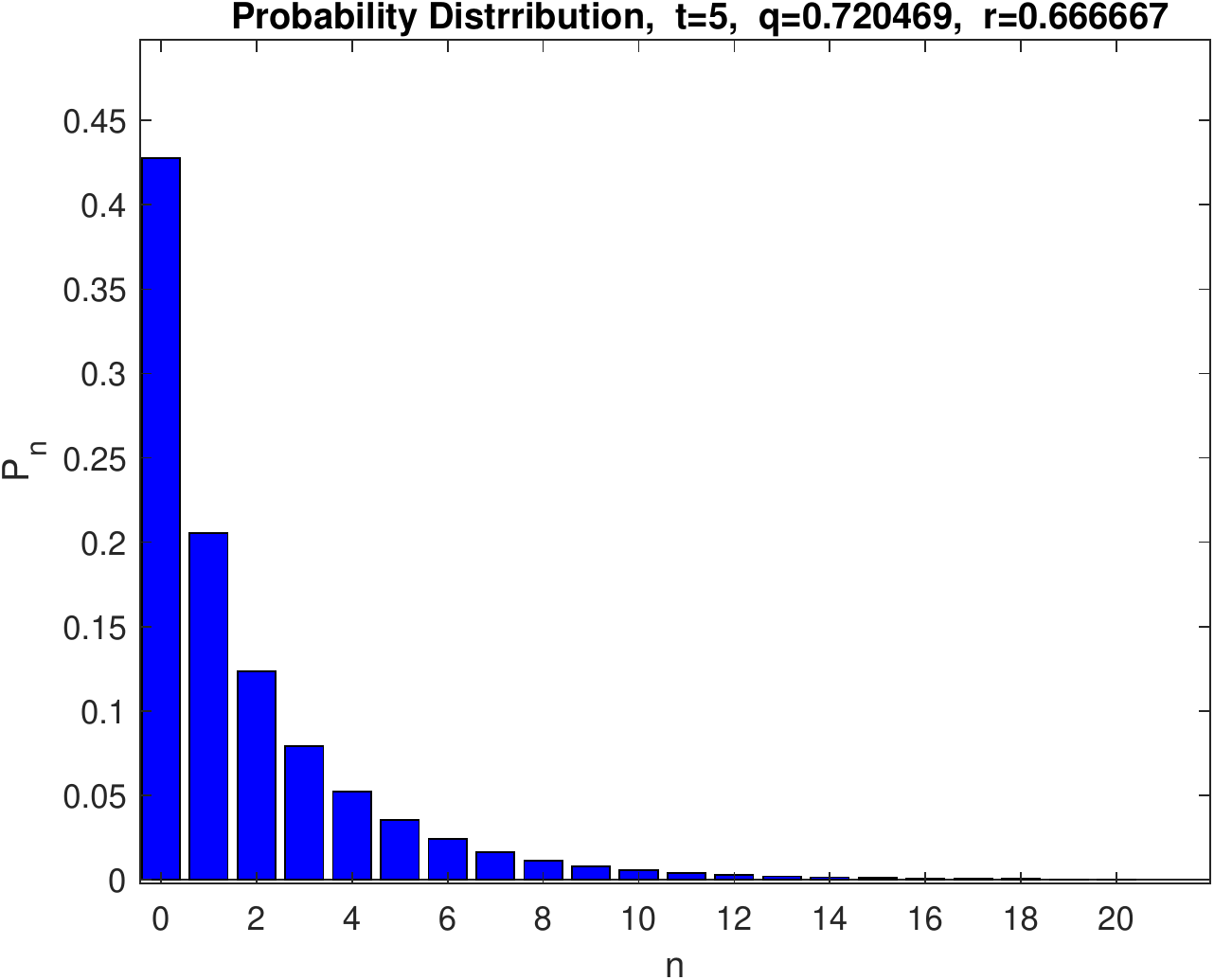}
  \caption{$P_n(t)$ at $t=5$.} 
  \label{fig:NBD_t=5}
  \end{minipage}
\end{figure}
\begin{figure}[hbt]
  \begin{minipage}{0.4\textwidth}
  \centering
  \includegraphics[scale=0.5]{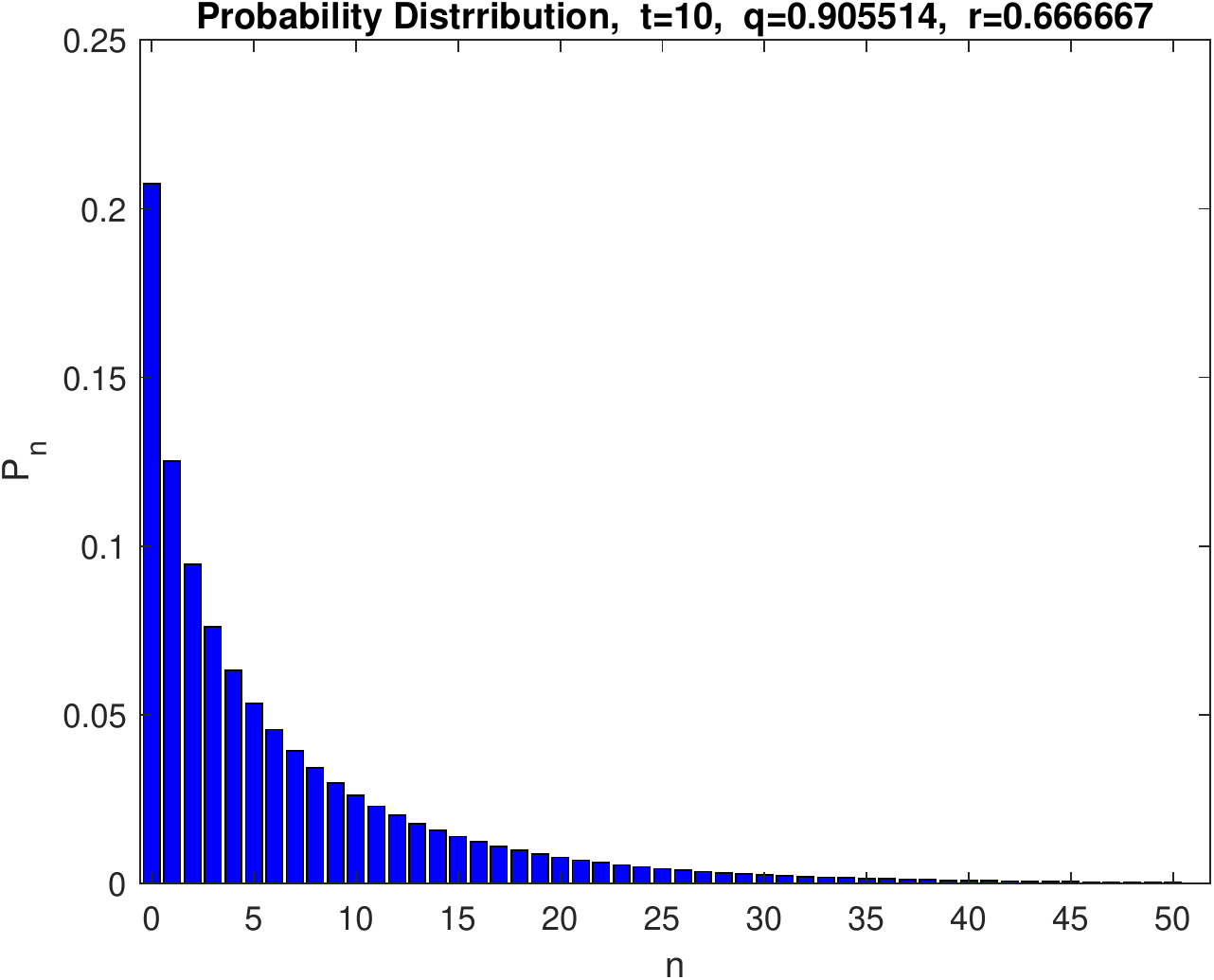}
  \caption{$P_n(t)$ at $t=10$.}\label{fig:NBD_t=10}
  \end{minipage}
  \begin{minipage}{0.4\textwidth}
  \centering
  \includegraphics[scale=0.5]{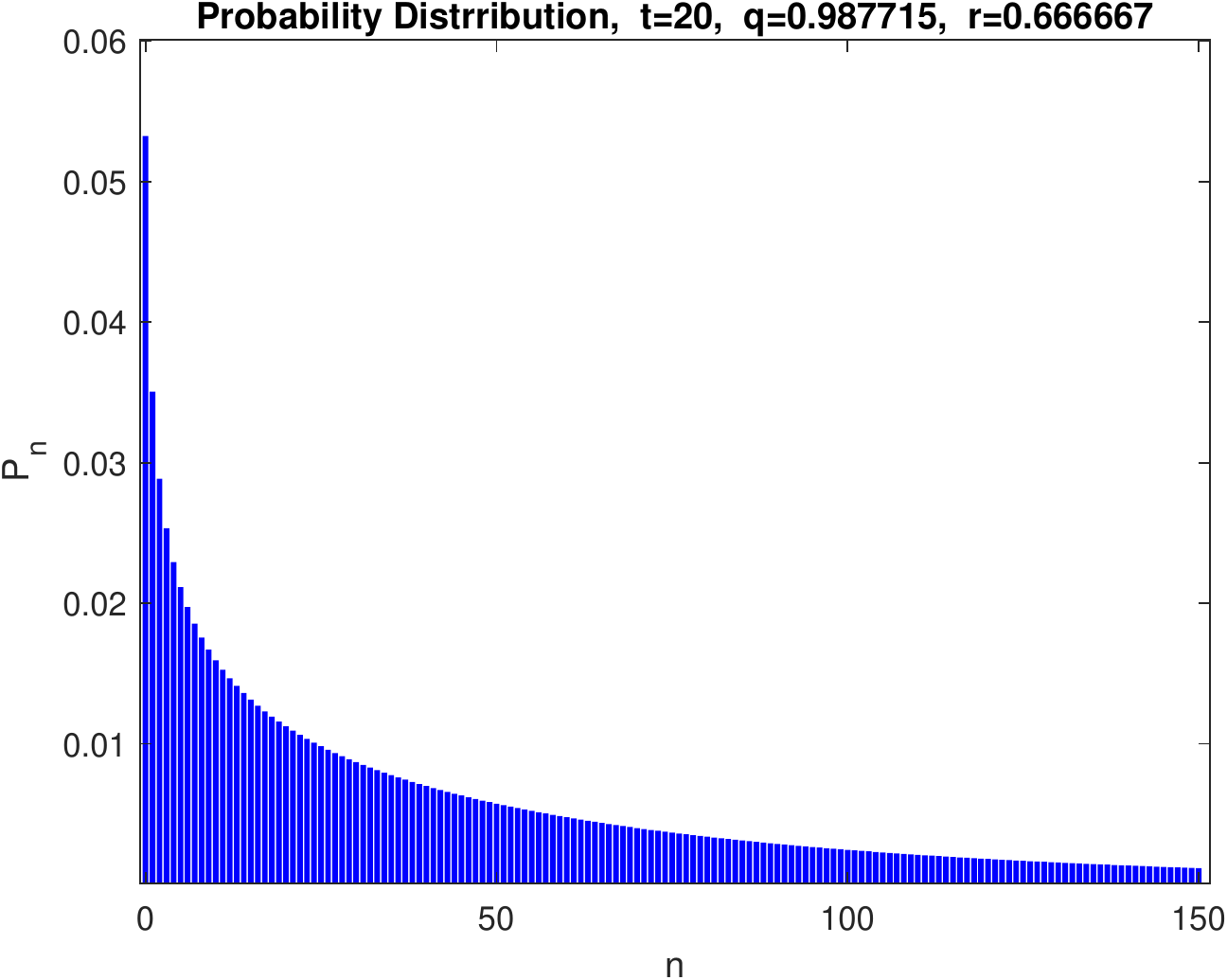}
  \caption{$P_n(t)$ at $t=20$.} \label{fig:NBD_t=20}
  \end{minipage}
\end{figure}
\begin{figure}[htb]
  \begin{minipage}{0.4\textwidth}
  \centering
  \includegraphics[scale=0.5]{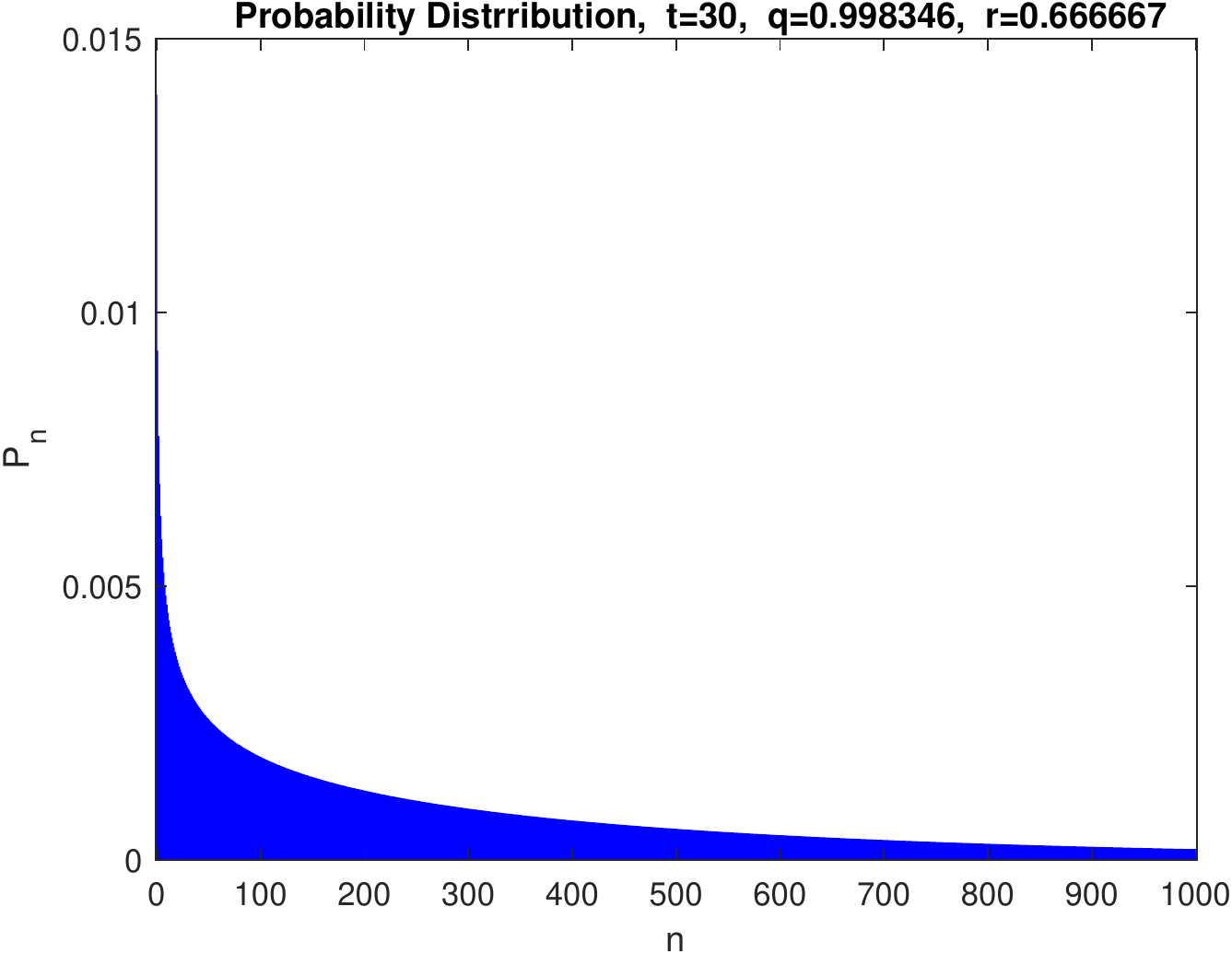}
  \caption{$P_n(t)$ at $t=30$.}\label{fig:NBD_t=30}
  \end{minipage}
  \begin{minipage}{0.4\textwidth}
  \centering
  \includegraphics[scale=0.5]{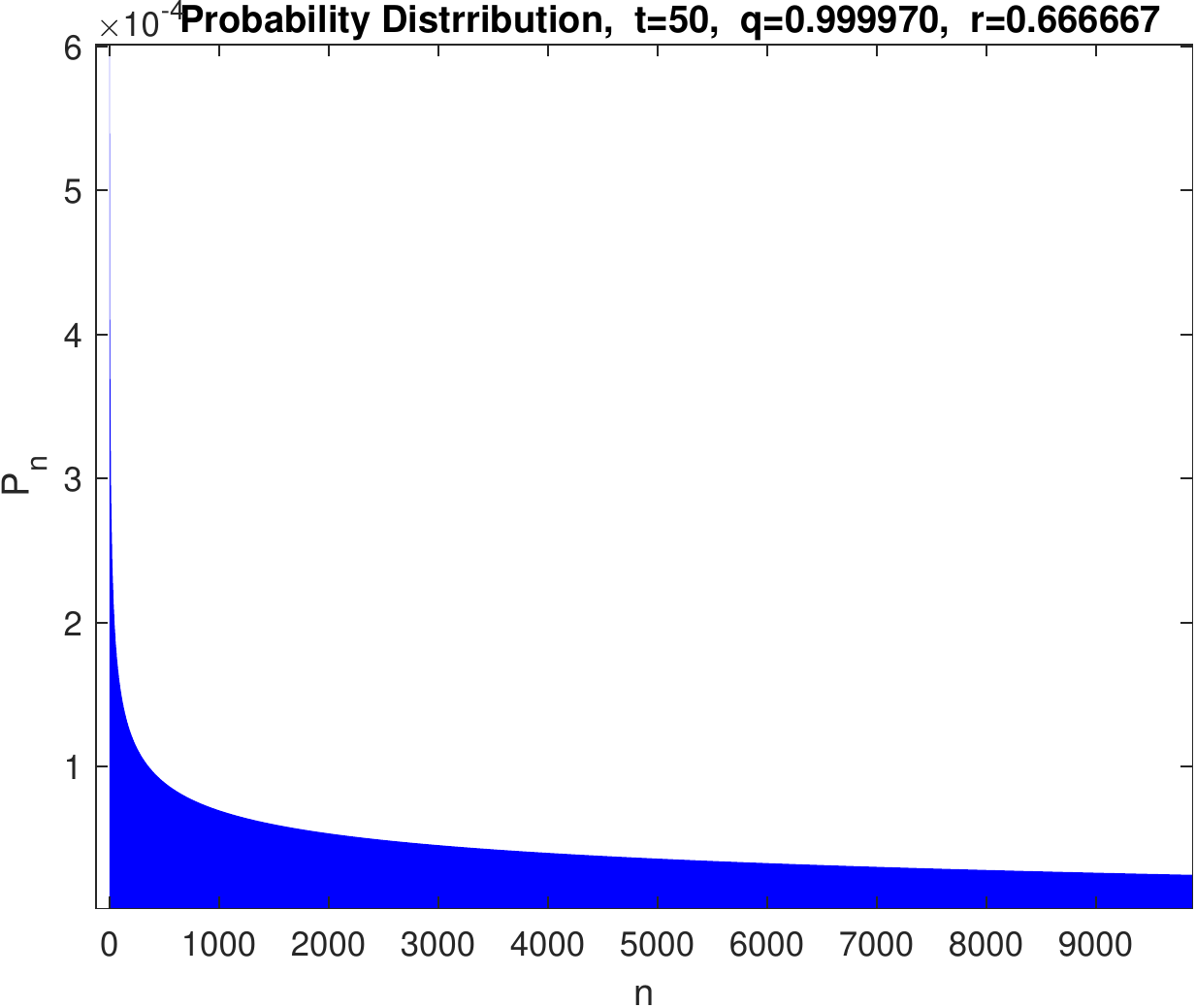}
  \caption{$P_n(t)$ at $t=50$.} \label{fig:NBD_t=50}
  \end{minipage}
\end{figure}
\begin{figure}[th]
  \begin{minipage}[b]{0.45\linewidth}
  \centering
  \includegraphics[width=\textwidth]{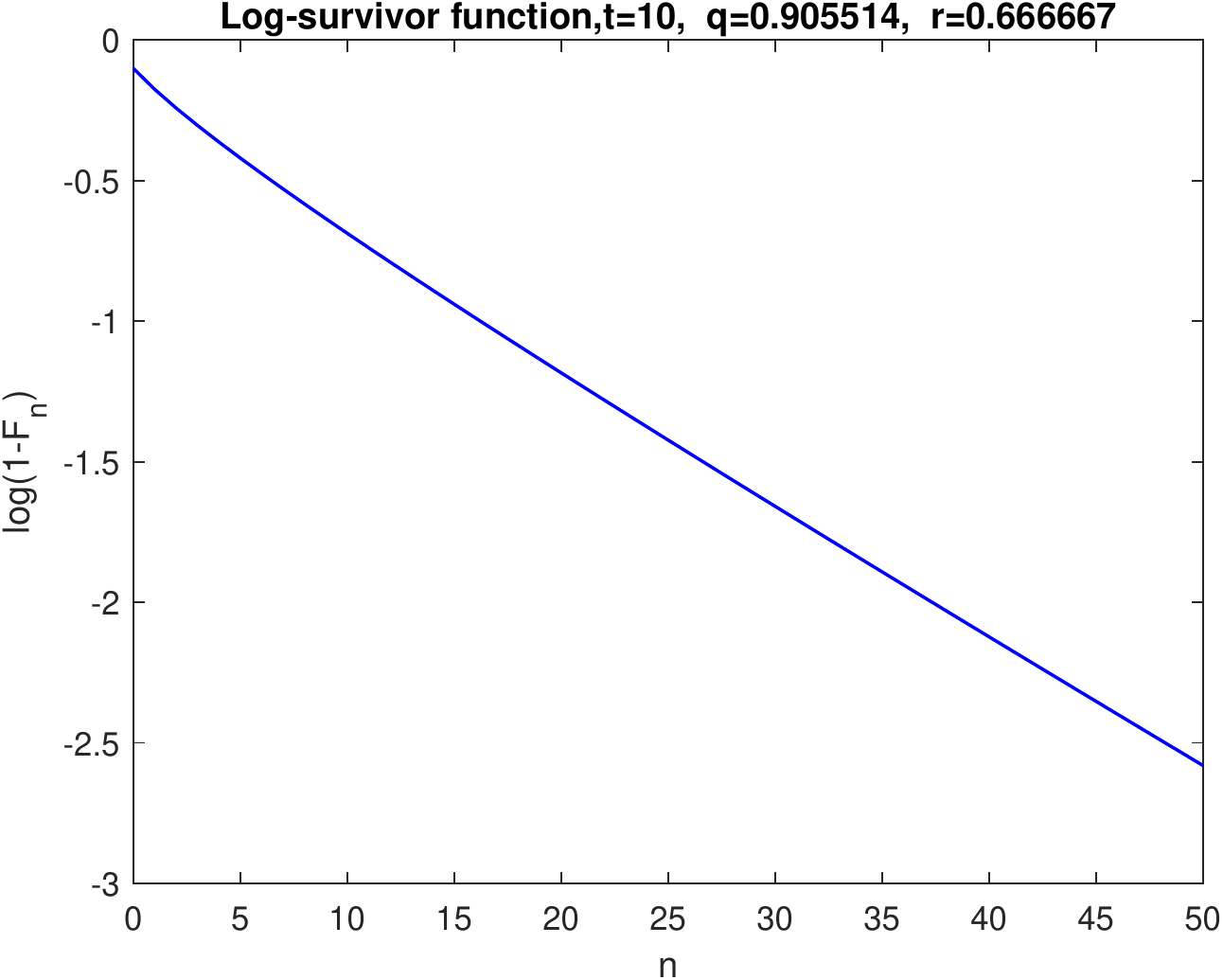}
  \caption{The log-survivor function plot of the distribution of $I(t)$ at $t=10$.}
  \label{fig:log_surv_t=10}
  \end{minipage}
  \hspace{0.5cm}
  \begin{minipage}[b]{0.45\linewidth}
  \centering
  \includegraphics[width=\textwidth]{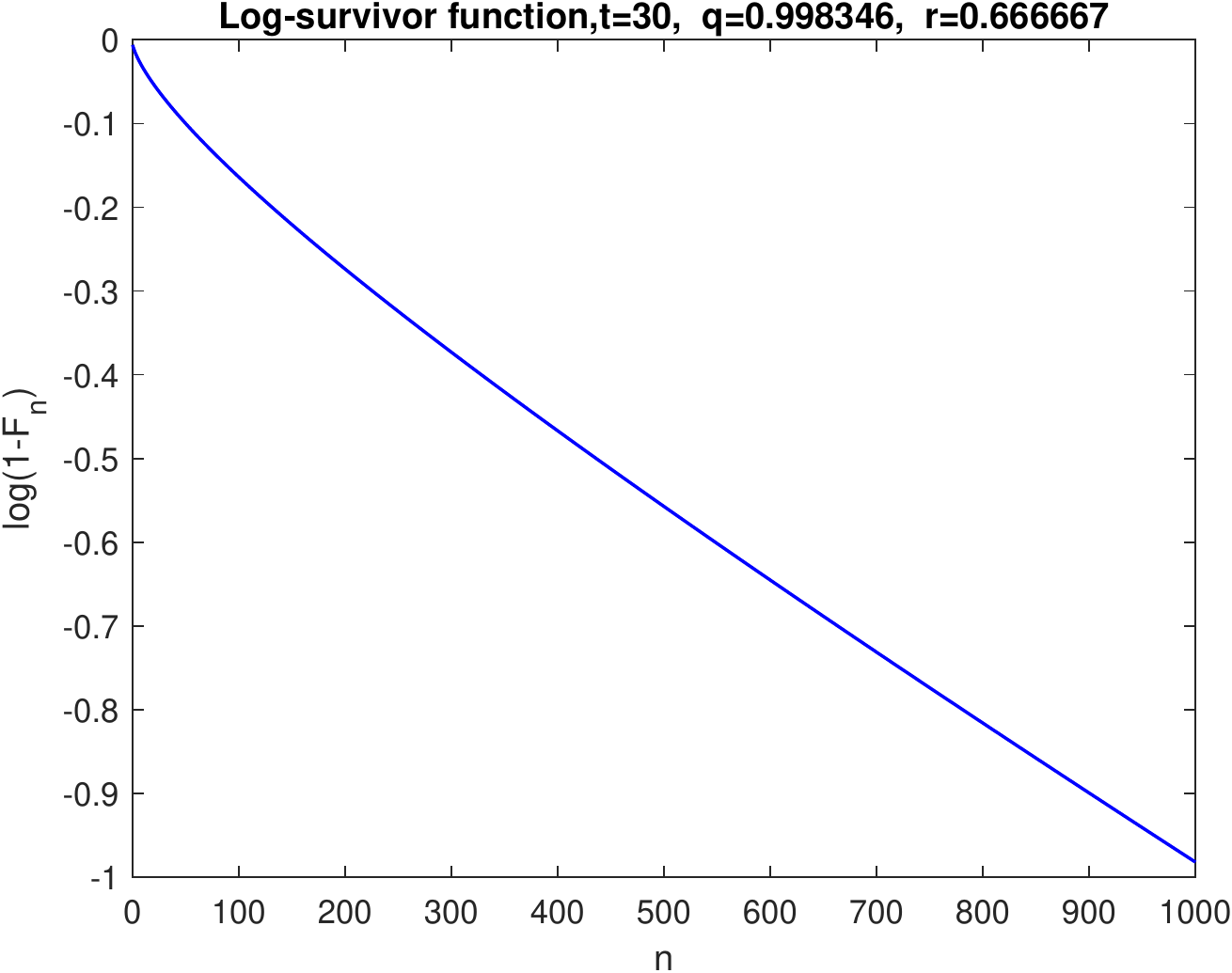}
  \caption{The log-survivor function plot of the distribution of $I(t)$ at $t=30$.} 
  \label{fig:log_surv_t=30}
  \end{minipage}
\end{figure}
%

As we observed in Example \ref{Example-NBD}, the PMF $P_n(t)$ looks similar to a slowly decaying geometric distribution for modest values of $t$, with a long tail, which gives a large variance. An often used measure of dispersion is the \textbf{\emph{coefficient of variation} (CV)}.  From (\ref{sigma-square-BDI}) we can derive the following expression: 
\begin{align}
\sigma^2_{\scriptscriptstyle BDI}(t)&=\oI^2(t)\frac{\lambda e^{at}-\mu}{\nu(e^{at}-1)}
=\frac{\oI^2(t)}{r\beta(t)},
\end{align}
from which we find the CV at time $t$ of the BDI process is given by
\begin{align}
c_{\scriptscriptstyle BDI}(t)=\frac{\sigma_{\scriptscriptstyle BDI}(t)}{\oI(t)}=\frac{1}{\sqrt{r\beta(t)}}.
\end{align}
It is evident that $\beta(t)$ quickly approaches unity from below, as $t$ increases\footnote{When $t$ is such  $at=7$, for instance, will make $\beta(t)$ within around 2\% off from the unity.  In our example of $a=0.2$[/day], $t=23$[days] will make $\beta(t)$ practically equal to unity.}
\begin{align}
\beta(t)\approx 1,~~\mbox{for}~~t\gg1,~~\mbox{and}~~\lim_{t\to\infty}\beta(t)=1.
\end{align}
Thus, we have arrive at the following proposition regarding the CV of the BDI process:
\begin{prop}
\textbf{\emph{The coefficient of variation}} of the BDI process rapidly converges to a constant $\sqrt{r^{-1}}$ for all $t$ such that $(\lambda-\mu)t\geq 7$:
\begin{equation}
\fbox{
\begin{minipage}{6cm}
\[ 
  \lim_{t\to \infty}c_{\scriptscriptstyle BDI}(t) =\sqrt{r^{-1}}, ~~\mbox{where}~~r=\frac{\nu}{\lambda}.
\]
\end{minipage}
} \label{limit-CV}
\end{equation}
\end{prop}

\noindent\\
\small{\textbf{Example 3:} \label{Example-covariance}
 
Consider the environment considered in Example \ref{Example_1-deterministic}, i.e., $\nu=0.2$/day, $\lambda=0.3$/day/infectious person, $\mu=0.1$/day/infectious person. Thus, $R_0=\lambda/mu=3$,  $a=\lambda-\mu=0.2.$, and $r=\frac{\nu}{\lambda}=2/3$.  Hence, $\sqrt{r^{-1}}= 1.225$.

The mean, variance and the standard deviation of $I(t)$ on the days $t=30$ and 50 are:
\begin{equation}
\begin{array}{lllll}
&\oI(0)=0,~~~~~&\mbox{Var}[I(0)] =0, ~~~~~&\sigma_I(0)=0, ~~~&c_I(0)=\mbox{undefined}, \nonumber\\
&\oI(30)=402.5,~~~~~&\mbox{Var}[N(30)]=243,411.9,~~~~&\sigma_N(30)=493.4, ~~& c_N(30)=1.226,\nonumber\\
&\oI(50)=22,025.5,~~&\mbox{Var}[N(50)]=727,706,000.9,~~&\sigma_N(50)=26,976.0,~~&c_N(50)=1.225.\nonumber\
\end{array}  
\end{equation} 
So, we have confirmed that the coefficient of variation is bounded from below by $\sqrt{r^{-1}}$=1.225. $\Box$
}
\subsection{Relation between the NBD and ``generalized" binomial distributions}

In the previous section we gave a formal definition of  the \emph{generalized} negative binomial distribution, by allowing the parameter $r$ to be a positive real number.  We now want to generalize the binomial distribution B($N, p$), by allowing $N$ to be a real number (positive or negative), and $p$ be any real number, positive or negative.
Thus, we cannot assign any probabilistic interpretation.

\begin{definition} \label{generalized-binomial}
Consider a generating function of $z$, which takes the form for real numbers, $p$ and $\alpha$.
\begin{align}
G(z)=(1-p + p~z)^\alpha.  \label{PGF-GB}
\end{align}
Then its inverse transform 
\begin{align}
f_n= {\alpha \choose n}p^n (1-p)^{\alpha -n},~~n=0, 1, 2, \ldots  \label{f_n}
\end{align}
is called a \textbf{\emph{generalized binomial distribution}}, denoted GB($n, p$).
\end{definition}

There is no guarantee that each $f_n$ non-negative, let alone $0\leq f_n\leq 1$, However, they add up to unity, because $G(1)=1$.  The main reason why we wish to generalize the binomial distribution is for the convenience of computing $P_n(t)$ for the BDI process and other birth-and-death processes (without immigration) discussed.

\begin{prop}[A generalized binomial distribution associated with a (generalized) negative binomial distribution]\label{prop-GB-GNB}

For a given (generalized) negative binomial distribution, whose PGF is given by
\begin{align}
G_{\scriptscriptstyle NB}(z)
=\left(\frac{1-q}{1-q z}\right)^r,  \label{PGF-GNB}
\end{align}
with the PMF
\begin{align}
P_n^{\scriptscriptstyle NB}={r+n-1\choose n} (1-q)^r q^n,~~n=0, 1, 2, \cdots,  \label{PMF-GNB}
\end{align}
its associated generalized binomial distribution is given by
\begin{align}
f_n={-r\choose n} p^n (1-p)^{-r-n}  \label{f_n-prop}
\end{align}
where $p$ and $q$ satisfy the relation
\begin{align}
(1-p)(1-q)=1, ~~or~~p=-\frac{q}{1-q}.  \label{p-and-q}
\end{align}
\end{prop}

\begin{proof}
First we compare the generalized binomial coefficients
\begin{align}
{-r\choose n}&=\frac{(-r)(-r-1)\cdots (-r-n+2)(-r-n+1)}{n!}\nonumber\\
&=\frac{(r+n-1)(r+n-2)\cdots (r+1)r(-1)^n}{n!}
=(-1)^n{r+n-1\choose n}\label{ninomial-coeff-formula}
\end{align}
We also find 
\begin{align}
p^n(1-p)^{-r-n}=\left(\frac{-q}{1-q}\right)^n (1-q)^{r+n}=(-1)^n q^n (1-q)^r \label{equivalence}
\end{align}
Thus, we have shown the equivalence between (\ref{PMF-GNB}) and (\ref{f_n-prop}).

An alternative way to prove this equivalence is to write the PGF (\ref{PGF-GNB}) using $p$. defined by (\ref{p-and-q}):
\begin{align}
\left(\frac{1-q}{1-q z}\right)^r=(1-p + p~z)^{-r}.  \label{q-and-p}
\end{align}
Recall \textbf{\emph{Newton's generalized binomial formula}}:
\begin{align}
(1+t)^a=1+{a\choose 1}t+{a\choose 2}t^2+\cdots =\sum_{j=0}^\infty{a\choose j}t^j.
\end{align}
Apply this formula to the RHS of (\ref{q-and-p}), by setting $a=-r$ and 
$t=\frac{pz}{1-p}$. The coefficient of the $z^n$ term is given by $f_n$ of (\ref{f_n})
$\Box$
\end{proof}

Now let us revisit the PGF (\ref{PGF-empty-1}), which we can write as
\begin{align}
G(z,t)=(1-p(t)+p(t)~z)^{-r},  \label{PFG-p(t)}
\end{align}
where
\begin{align}
p(t)=-\frac{\lambda e^{at}-\mu}{a}.  \label{def-p(t)}
\end{align}
By referring to Proposition \ref{prop-GB-GNB} we find that $p(t)$ and $\beta(t)$ are related by 
\begin{align}
p(t)=-\frac{\beta(t)}{1-\beta(t)}.
\end{align}
Thus, by applying Newton's formula, we can obtain the following expression for $P_n(t)$:
\begin{align}
P_n(t)={-r\choose n}(1-p(t))^n p(t)^{-r-n},~~n=0, 1, 2, \cdots.  \label{P_n-GB}
\end{align}
Using the binomial coefficient formula and the equivalence (\ref{equivalence}), i.e.,
\begin{align}
p^n(1-p)^{-r-n}=(-1)^n \beta^n(1-\beta)^r,
\end{align}
we see that $P_n(t)$ given by (\ref{P_n-GB}) is equivalent to (\ref{time-dependent-sol}).

\subsection{The NBD as a Compound Poisson Distribution} \label{subsec-Compound Poisson}
In this section we will investigate how to represent the BDI process $I_{\scriptscriptstyle BDI}(t)$ as a compound Poisson process. The negative binomial distribution is \textbf{\emph{infinitely divisible}} (see e.g., Feller \cite{feller:1968}), and hence can be represented as a compound Poisson process.   First we show this representation by considering a special case of the BDI process where no death occurs.

\subsubsection{Birth-and-immigration (BI) process}\label{subsec-BI}

Let us consider a process which has not been discussed, to the best of our knowledge, in the literature. That is, the pure birth process, accompanied by immigration.  Let us call this process a \textbf{\emph{birth-and-immigration process}}, a BI process for short. 

The PGF of this process can be readily found by setting $\mu=0$ in (\ref{PGF-general-BDI}), and after some manipulation, we obtain 
\begin{align}
G_{\scriptscriptstyle BI}(z,t)
&=\left(\frac{1}{e^{\lambda t}-(e^{\lambda t}-1)z}\right)^{r}
\left(\frac{z}{e^{\lambda t}-(e^{\lambda t}-1)z}\right)^{I_0}
=z^{I_0}\left(\frac{1-\beta_0(t)}{1-\beta_0(t) z}\right)^{r+I_0},
 \label{PGF-general-BI}
\end{align}
where 
\begin{align}
\beta_0(t)=1-e^{-\lambda t}. \label{def-beta_0}
\end{align}
Not surprisingly, this $\beta_0(t)$ can be obtained by setting $\mu=0$ in $\beta(t)$ of (\ref{def-beta(t)}). Thus, the BI process $I_{\scriptscriptstyle BI}(t)$ is distributed according to the negative binomial distribution NB($r+I_0, \beta_0(t)$) shifted by $I_0$. Hence, we find the PMF at time $t$ as given below:
\begin{align}
P_n^{\scriptscriptstyle BI}(t)=e^{-(r+I_0)\lambda t}{n-I_0+r-1\choose n-I_0}
\left(1-e^{-\lambda t}\right)^{n-I_0},~~n=I_0, I_0+1, I_0+2, \cdots.
\end{align}
If, $I_0=0$, the PMF reduces to NB($r, \beta_0(t)$).  By noting $r\lambda=\nu$, we find the first few PMF values as follows:
\begin{align}
P_0^{\scriptscriptstyle BI:0}(t)&= e^{-\nu t}, \nonumber\\
P_1^{\scriptscriptstyle BI:0}(t)&= e^{-\nu t}r                     \left(1-e^{-\lambda t}\right),  \nonumber\\
P_2^{\scriptscriptstyle BI:0}(t)&= e^{-\nu t}r\frac{(r+1)}{2!}     \left(1-e^{-\lambda t}\right)^2, \nonumber\\
P_3^{\scriptscriptstyle BI:0}(t)&= e^{-\nu t}r\frac{(r+2)(r+1)}{3!}\left(1-e^{-\lambda t}\right)^3, \nonumber
\end{align}
where $P_0^{\scriptscriptstyle BI:0}(t)$ represents the \textbf{\emph{idle period}} of the system which ends upon the Poisson arrival.  This initial idle period is the only time when the system is empty.  Since there is no death or departure of any kind, the population monotonically increases as the time elapses.

As we defined in Definition \ref{def-A(t)-etc} of Section \ref{subsec-Expectation}, $A(t)$ is the cumulative count of  infected and infectious persons, arriving from the outside up to time $t$.  Let us count them as $c=1, 2, \cdots, A(t)$.  Consider a ``countable attribute" associated with each arrival and denote it $N_c$. Assume that the $N_c$'s are independent identically distributed (i.i.d.) RVs.  Kobayashi and Konheim \cite{kobayashi-konheim:1975} discuss use of a\emph{compound process} in the context of data communication system, where each arrival of a data packet carries $N_c$ [data units], e.g., [bytes]. Then, the sum 
\begin{align}
S(t)=\sum_{c=1}^{A(t)} N_c  \label{sum}
\end{align}
should be a quantity of interest in determining, e.g., the buffer size of a statistical multiplexer needed to keep the probability of buffer overflow below some prescribed level.
 
In our problem at hand, each infected arrival from the outside becomes an ``ancestor" who produces many ``descendants" over multiple generations, who are all ``infected" persons and constitute the members of the ``$c$ clan."  Such a process is called a \textbf{\emph{branching process}}.  In our model, each infected individual, whether an ancestor or a descendant, produces a new descendant at rate $\lambda$ [person/unit-time].

Let $A_j(t)$ be the number of arrivals in $(0, t]$ and having $j$ \emph{descendants} at time $t$.  The case $j=0$ does not exist, since there is no death, hence each clan includes at least its ancestor, such that
\begin{align}
\sum_{j=1}^\infty A_j(t)=A(t), ~~t\geq 0.
\end{align}
Branson \cite{branson:1991}, who ascribes the original idea to Karin and McGregor \cite{karlin-mcgregor:1967},
shows that the RVs $A_j(t)$'s are independent Poisson variables with mean $m_j(t)$ where
\begin{align}
m_j(t)&=E[A_j(t)]=\frac{\nu}{\lambda}\frac{\xi^j}{j}, \label{Branson}
\end{align}
where $\xi$ is a function of $\lambda$ and $t$.

The above formula (\ref{Branson}) leads to the idea that the probability distribution that each ancestor should produce $j$ descendants (i.e., secondary, tertiary infectees, etc.) whose probability distribution takes the form
\begin{align}
\mbox{Pr}[N_c(t)=j]=P^{\scriptscriptstyle N_c}_j(t)\propto \frac{\beta_0(t)^j}{j}, ~~j\geq 1. \label{conjecture}
\end{align}
This distribution is 
the {\emph{logarithmic distribution} of (\ref{Fisher}). The PGF of $N_c(t)$ is therefore given by
\begin{align}
G_{\scriptscriptstyle N_c}(z)&=\frac{\ln(1-\beta_0(t)~z)}{\ln (1-\beta_0(t))}
=\frac{\ln(1-(1-e^{-\lambda t}) z)}{-\lambda t}
 \label{PGF-N_c}
\end{align}

Note that $N_c(t)$ is the number of the family members at time $t$ of the $c$th clan, with the ancestor itself included. Since we are assuming no death, the sum of these numbers over all $A(t)$ ancestors should be the total population $I(t)$ at time $t$:
\begin{align}
I(t)=\sum_{c=1}^{A(t)} N_c(t).  
\end{align}
Then, the PGF of $I(t)$ can be expressed as
\begin{align}
G(z,t)&=\Ex\left[z^{I(t)}\right]=\Ex\left[z^{\sum_{c=1}^{A(t)}N_c(t)}\right]  
= \Ex\left[\left(\Ex[z^{N_c}]\right)^{A(t)}\right]
= \Ex\left[\left(G_{\scriptscriptstyle N_c}(z, t)\right)^{A(t)}\right]\nonumber\\
&= G_A\left(G_{\scriptscriptstyle N_c}(z)\right)
\end{align}

The PGF of the Poisson arrival process of rate $\nu$ is given by 
\begin{align}
G_{\scriptscriptstyle A}(z,t)=\exp(\nu t(z-1)). \label{PGF-Poisson}
\end{align}
Thus, we find the PGF of the BI process, a compound Poisson process,  is
\begin{align}
G(z,t)&=G_{\scriptscriptstyle A}(G_{\scriptscriptstyle N_c}(z))=\exp\left\{\nu t\left(
\frac{\ln(1-\beta_0(t))}{\ln(1-\beta_0(t)~ z)}-1\right)\right\}\nonumber\\
&=\left(\frac{1-\beta_0(t)}{1-\beta_0(t)~ z}\right)^r, ~~\mbox{where}~~r=\frac{\nu}{\lambda},
\label{PGF-as-compound}
\end{align}
which is the PGF (\ref{PGF-general-BI}), where $I_0=0$.  Thus, we have shown that the Poisson process with rate $\nu$ ends up with the negative binomial distributed process NB($r, \beta_0(t)$).  It is interesting that the infection process which forms a \textbf{\emph{branching process}} where an infectious person infects a susceptible person at rate $\lambda$ [persons/time unit] will result in the negative binomial distribution. The \emph{positive feedback loop} of creating its descendants exhibits the negative binomial distributed process.

\subsubsection{Birth-and-death with immigration (BDI) process as a compound Poisson process}

Now we extend our argument of the previous section to the BDI process with the initial condition $I_0=0$, whose PGF is given by (\ref{PGF-empty-1}).
\begin{align}
G(z,t)=\left(\frac{1-\beta(t)}{1-\beta(t)z}\right)^r,  \label{PGF-empty-2}
\end{align}
where $\beta(t)$ is defined by (\ref{def-beta(t)}):
\begin{align}
\beta(t)=\frac{\lambda(e^{a t}-1)}{\lambda e^{a t}-\mu}. \label{def-beta(t)-2}
\end{align}
The first few PMF values can be found from (\ref{time-dependent-sol}) as
\begin{align}
P_0^{\scriptscriptstyle BDI:0}(t)&=B(t),~~\mbox{where}~~B(t)=(1-\beta(t))^r=\left(\frac{a}{\lambda e^{at}-\mu}\right)^r,\nonumber\\
P_1^{\scriptscriptstyle BDI:0}(t)&=B(t)r\beta(t), \nonumber\\
P_2^{\scriptscriptstyle BDI:0}(t)&=B(t)r\frac{(r+1)}{2!}\beta(t)^2, \nonumber\\
P_3^{\scriptscriptstyle BDI:0}(t)&=B(t)r\frac{(r+2)(r+1)}{3!}\beta(t)^3. 
\end{align}
The function $\beta(t)$ is between 0 and 1, and approaches 1 as $t\to\infty$. If $\nu\to 0$, then $r\to 0$ as well and $B(t)\to 1$.  Since $B(t)r$ does not depend on $n$, we see again the \textbf{\emph{Fisher series}}:
\begin{align}
\mbox{Pr}[N_c(t)=j]=P^{\scriptscriptstyle N_c}_j(t)\propto \frac{\beta(t)^j}{j}, ~~j\geq 1. \label{conjecture-2}
\end{align}
Then by replacing $\beta_0(t)=1-e^{-\lambda t}$ by $\beta(t)=\frac{\lambda(e^{at}-1)}{\lambda e^{at}-\mu}$ the entire mathematical steps of the BI process case carry over to the case of the BDI process, except for the last step (\ref{PGF-as-compound}), where the parameter $r$ needs to be defined as
\begin{align}
r=\frac{-\nu t}{\ln\left(\frac{\lambda -\mu}{\lambda e^{at}-\mu}\right)}.
\end{align}
If $a=\lambda-\mu>0$, then for $t\gg 1$, the following approximation holds
\begin{align}
r\approx \frac{-\nu t}{-at}= \frac{\nu}{\lambda}\frac{\lambda}{\lambda-\mu}.
\end{align}
Therefore, it seems clear that we need to to choose a Poisson process other than the one with mean $\nu t$.  Let $\nu(t)$ be the mean of the Poisson process which is to be compounded with $N_c(t)$ whose PGF is
\begin{align}
G_{\scriptscriptstyle BDI:0}(t)=\left(\frac{1-\beta(t)}{1-\beta(t)~z}\right)^r.
\end{align}
Then it should be clear from the above discussion that $\nu(t)$ must be given by
\begin{align}
\exp(\nu(t))&=\left(\frac{\lambda e^{at}-\mu}{\lambda-\mu}\right)^r 
=\exp(art)\left(\frac{\lambda-\mu e^{-at}}{\lambda -\mu}\right)^r  \nonumber\\
&=\exp(\nu t)\left(\frac{\lambda-\mu e^{-at}}{\lambda -\mu}\right)^r.
\end{align} 
Thus,
\begin{align}
\nu(t)=\nu t +r\ln \left(\frac{\lambda-\mu e^{-at}}{\lambda -\mu}\right).
\end{align}
Thus, we have shown that the BDI process as a compound Poisson process with rate $\nu(t)$.

\section{Concluding Remarks}

At the beginning of Section \ref{sec-stochastic_model} we showed the results of twelve simulation runs of the BDI process to provoke the reader's interest.  Hopefully you have found an answer to the first question. How about the second question? If you are an avid reader, you should have figured out the answer by now. In Part II of this report under preparation \cite{kobayashi:2020b}, we will present more simulation results, which should help the reader find a definitive answer to the second question.  

In this report, we focused on the analysis of $I(t)$, the infected population. Although we presented the stochastic means of $B(t)$ (the population of the secondary infections) and $R(t)$ (the recovered/dead population), we defer a full analysis of these processes to Part II.

We will also show, both by analysis and simulation, how the infection process can be controlled by changing the values of $\lambda, \mu$ and $\nu$ at certain points in time.  Increasing the so-called \emph{social distance} would decrease the value of $\lambda$.  The value of $\mu$ can be increased by improving medical treatments to speed up the recovery process. Another possible option is to reduce the size of the susceptible population  significantly by exposing the young and healthy susceptible to the disease so that they  become immune to the disease. The quantitative model such as ours should help policy makers and health officials to make judicious choices of these options by assessing the effects and costs of various options available to them.  In Part II (or III), we will also discuss how we should estimate the model parameters from empirical data so that the forecast can be made as reliable as possible.

\appendix

\numberwithin{equation}{section}

\section{Derivation of PGF (\ref{PGF-solution})}\label{Appendix A}
The partial differential equation (PDE) (\ref{def-PGF}) can be solved for PGF $G(z,t)$ by Lagrange's method with auxiliary differential equations.  We write (\ref{def-PGF}) in the following form of a {\it planar differential equation} (see e.g., \cite{gross-harris:1985}, \cite{kobayashi-ren:1992}, \cite{ren-kobayashi:1995}, \cite{kobayashi-mark:2008}, pp. 600-603):
\begin{align}
p\frac{\partial G}{\partial t}+q\frac{\partial G}{\partial z}=r, \label{planar-dif-eq}
\end{align}
where $G$ denotes $G(z,t)$ defined above, and $p, q$ and $r$ are, in general,
functions of $t, z,$ and $G$. 

The solution $G=f(t, z)$ represents a surface in the $t-z-G$ space. The normal vector to the surface at any point $(t_0, z_0, G_0)$ is perpendicular to the line through this point, whose direction numbers are the values of $p, q$ and $r$ evaluated at this point, which we denote by $p_0, q_0$ and $G_0$. This line has the equations
\begin{align}
\frac{t-t_0}{p_0}=\frac{z-z_0}{q_0}=\frac{G-G_0}{r_0}.
\end{align}
Hence, at each point $(p, q, G)$ of the surface, there is a normal vector whose direction numbers, $dt, dz$ and $dG$, satisfy the following equation:
\begin{align}
\frac{dt}{p}=\frac{dz}{q}=\frac{dG}{r}.\label{normal_vector}
\end{align}

In our problem at hand, we set
\begin{align}
p&=1,\\
q&=-(\lambda z-\mu)(z-1),\\
r&=\nu(z-1)G,
\end{align}
Then (\ref{normal_vector}) becomes
\begin{align}
\frac{dt}{1}= -\frac{dz}{(\lambda z-\mu)(z-1)}=\frac{dG(z,t)}{\nu(z-1)G(z,t)}.\label{normal_vector-1}
\end{align}
which leads to the following two independent solutions:
\begin{align}
\left(\frac{z-1}{\mu-\lambda z}\right)e^{a t}=C_1, \mbox{where}~~a=\lambda-\mu, \label{C_1}
\end{align}
and
\begin{align}
G(z,t)(\mu-\lambda z)^{r}=C_2, ~~\mbox{where}~~r=\frac{\nu}{\lambda} \label{C_2}
\end{align}
where $C_1$ and $C_2$ are integration constants.
We write the functional relations between $C_1$ and $C_2$ as
\begin{align}
C_2=f(C_1), \label{C_2:C_1},
\end{align}
which, together with (\ref{normal_vector-1}), implies the following general solution:
\begin{align}
G(z,t)=(\mu-\lambda z)^{-r}~
f\left(\left(\frac{z-1}{\mu-\lambda z}\right)e^{\lambda t}\right). \label{G-f}
\end{align}
From the initial condition (\ref{initial}), we find
\begin{align}
G(z,0)=(\mu-\lambda z)^{-r} f\left(\frac{z-1}{\mu-\lambda z}\right)=z^{I_0}.
\end{align}
from which we find
\begin{align}
f\left(\frac{z-1}{\mu-\lambda z}\right)= z^{I_0}(\mu-\lambda z)^{r}.
\end{align}
We introduce a new variable $y$ by
\begin{align}
\frac{z-1}{\mu-\lambda z}=y,~~\mbox{i.e.}, ~~
z=\frac{1+\mu y}{1+\lambda y}.
\end{align} 
From the last two equations, we can determine the functional form of $f$:
\begin{align}
f(y)= \left(\frac{1+\mu y}{1+\lambda y}\right)^{I_0} 
\left(\frac{\mu-\lambda}{1+\lambda y}\right)^{r}.
\label{func-f}
\end{align}

By substituting the last equation into (\ref{G-f}), we obtain
\begin{align}
G(z,t)=(\mu -\lambda z)^{-r}
f\left(\left(\frac{z-1}{\mu-\lambda z}\right)e^{a t}\right)
\end{align}
By noting
\begin{align}
\frac{1+\mu y}{1+\lambda}=\frac{\mu-\lambda z + \mu(z-1) e^{a t}}{\mu -\lambda z +\lambda (z-1) e^{a t}}
\end{align}
and 
\begin{align}
\frac{\mu-\lambda}{1+\lambda y}&=\frac{(\mu-\lambda)(\mu-\lambda z)}{\mu-\lambda z+\lambda (z-1)e^{a t}}
\end{align}
we finally obtain the PGF of $I(t)$
\begin{align}
G(z,t)&=(\mu -\lambda z)^{-r} 
\left(\frac{(\mu-\lambda)(\mu-\lambda z)}{\mu-\lambda z+\lambda(z-1)e^{a t}}\right)^{r}
\left(\frac{\mu-\lambda z+\mu(z-1)e^{a t}}{\mu-\lambda z + \lambda (z-1) e^{a t}}\right)^{I_0}\nonumber\\
&=\left(\frac{\mu-\lambda}{\mu-\lambda z+\lambda(z-1)e^{a t}}\right)^{r}
\left(\frac{\mu-\lambda z+\mu(z-1)e^{a t}}{\mu-\lambda z + \lambda (z-1) e^{a t}}\right)^{I_0}\nonumber\\
 \label{PGF-general-BDI}
\end{align}
Note that the PGF is given as a product of two PGFs, the second being the PGF of the birth-and-death process without immigration. 

If $I_0=0$, i.e., if the system is initially empty, the above
PGF reduces to
\begin{align}
G(z,t)&=\left(\frac{a}{\lambda e^{a t}-\mu -\lambda(e^{a t}-1)~z}
\right)^r, ~~\mbox{when}~~I(0)=I_0=0.  \label{PGF-empty}
\end{align}

\section*{Acknowledgments}
\addcontentsline{toc}{section}{Acknowledgments}
I thank Prof. Brian L. Mark of George Mason University for his valuable suggestions and help during the course of this study. He has read this manuscript carefully and made numerous editorial suggestions. Were it not for his patient help, I would have taken a lot more time in debugging my MATLAB simulation programs. Prof. Hideaki Takagi of the University of Tsukuba kindly shared with me his unpublished lecture note on the birth-and-death processes \cite{takagi:2007}. My excitement of having obtained a closed form solution of the time-dependent PMF of the BDI process was short lived, when he showed me his lecture note and informed me of the existence of the book by Bailey \cite{bailey:1964}. Dr. Linda Zeger read the first draft of this report and gave me valuable suggestions to improve the presentation.

\bibliographystyle{ieeetr}
\bibliography{infections}

\end{document}